\documentclass[]{article}

\usepackage[utf8]{inputenc}
\usepackage[english]{babel}

\usepackage{amsmath}
\usepackage{amsthm}
\usepackage{amssymb}
\usepackage{graphicx}
\usepackage{subfig}
\usepackage{natbib}
\usepackage{authblk}
\usepackage{algorithm}
\usepackage{algpseudocode}

\usepackage[hidelinks]{hyperref}
\newtheorem{thm}{Theorem}

\newtheorem{lem}{Lemma}

\newtheorem{example}{Example}

\newtheorem{innercustomthm}{Theorem}
\newenvironment{customthm}[1]
{\renewcommand\theinnercustomthm{#1}\innercustomthm}
{\endinnercustomthm}

\newtheorem{innercustomexm}{Example}
\newenvironment{customexm}[1]
{\renewcommand\theinnercustomexm{#1}\innercustomexm}
{\endinnercustomexm}

\title{Agent Failures in All-Pay Auctions}

\author[1]{Yoad Lewenberg\thanks{yoadlew@cs.huji.ac.il}}
\author[2]{Omer Lev\thanks{omerl@cs.toronto.edu}}
\author[3]{Yoram Bachrach\thanks{yobach@microsoft.com}}
\author[1]{\\Jeffrey S. Rosenschein\thanks{jeff@cs.huji.ac.il}}

\affil[1]{The Hebrew University of Jerusalem, Israel}
\affil[2]{University of Toronto, Canada}
\affil[3]{Microsoft Research, Cambridge, United Kingdom}

\date{}
\begin{document}

\maketitle

\begin{abstract}
All-pay auctions, a common mechanism for various human and agent interactions, suffers, like many other mechanisms,
from the possibility of players' failure to participate in the auction.
We model such failures,
and fully characterize equilibrium for this class of games, we
 present a symmetric equilibrium and show that under some conditions the equilibrium is unique.
We reveal various properties of the equilibrium, such as the lack of influence of the most-likely-to-participate player on the behavior of the other players.
We perform this analysis with two scenarios: the sum-profit model, where the auctioneer obtains the sum of all submitted bids,
and the max-profit model of crowdsourcing contests,
where the auctioneer can only use the best submissions and thus obtains only the winning bid.

Furthermore, we examine various methods of influencing the probability of participation such as the effects of misreporting one's own probability of participating,
and how influencing another player's participation chances changes the player's strategy.
\end{abstract}

\section{Introduction}

Auctions have been the focus of much research in economics, mathematics and computer science,
and have received attention in the AI and multi-agent communities as a significant tool for resource and task allocation.
Beyond explicit auctions, as performed on the web (e.g., eBay) and in auction houses, auctions also model various real-life situations
in which people (and machines) interact and compete for some valuable item.
For example, companies advertising during the U.S.~Superbowl are, in effect,
bidding to be one of the few remembered by the viewer, and are thus putting in tremendous amounts of money in order to create a memorable and unique event for the viewer,
overshadowing the other advertisers.

A particularly suitable auction for modeling various scenarios in the real world is the \emph{all-pay auction}.
In this type of auction, all participants announce bids, and all of them pay those bids,
while only the highest bid wins the product. Candidates applying for a job are, in a sense,
participating in such a bidding process, as they put in time and effort preparing for the job interview,
while only one of them is selected for the job. This is a \emph{max-profit} auction,
as the auctioneer (employer, in this case), receives only the top bid. In comparison,
a workplace with an ``employee of the month'' competition is a \emph{sum-profit} auctioneer,
as it enjoys the fruits of all employees' labour, regardless of who won the competition.

The explosion in mass usage of the web has enabled many more all-pay auction-like interactions,
including some involving an extremely large number of participants. For example, various crowdsourcing contests,
such as the Netflix challenge,
involve many participants putting in effort,
with only the best performing one winning a prize. Similar efforts can be seen throughout the web,
in TopCoder.com, Amazon Mechanical Turk, Bitcoin mining and other frameworks.

However, despite the research done on all-pay auctions in the past few years~\citep{DV09,CHS12,LPBR13},
some basic questions about all-pay auctions remain --- in a full information setting,
any equilibrium has bidders' expected profit at $0$, raising, naturally,
the question of why bidders would participate.

Several extensions to the all-pay auction model have been suggested in order to answer this question.
For example, \cite{LPBR13} showed that allowing bidders to collude enables the cooperating bidders to have a positive expected profit, at the expense of others bidders or the auctioneer.
This paper addresses this question by suggesting a model in which the bidders have a positive expected profit.

Furthermore, in all-pay auctions,
the number of bidders is a crucial information in order to bid according to the equilibrium~\citep{ref1}.
Hence, the number of participants must be known to the bidder. 
We suggest a relaxation of this assumption by allowing the possibility of bidders' failure, that is, there is a probability that a bidder will not be able to participate in the auction.
Therefore, we assume that the number of \emph{potential} bidders and the failure probability of every bidder are common knowledge, but not the exact number of participants.

As most large-scale all-pay auction mechanisms have variable participation,
we believe this helps capture a large family of scenarios,
particularly for online, web-based, situations and the uncertainty they contain.
We propose a symmetric equilibrium for this situation, we show when it is unique and prove its various properties.
Somewhat surprisingly, allowing failures makes the expected profit for bidders positive,
justifying their participation.

We start by reviewing related work in Section~\ref{sec:relW}. We then introduce the model with and without failures.
In Section~\ref{sec:ownF}, we first examine the case where each bidder has a different failure probability. 
Next, in Section~\ref{sec:sbtg}, we study the potential manipulations possible in this model, such as announcing a false probability (e.g., saying that you will put all your time into a TopCoder.com project) and changing the probability of others (e.g., sabotaging their car). 
Finally, in Section~\ref{sec:uni}, 
as calculations in this general case are complex,
we examine situations where bidders have the same failure probability (as is possible when weather or web server failure, for example, are the main determinant of participation), enabling us to detail more information about the equilibrium in this state.
In those situations we examine the effectiveness of changing the participation probability for all the bidders (e.g., convincing a deity to make it snow or attacking the server).

\section{Related Work}\label{sec:relW}
Initial research on all-pay auction was in the political sciences, modeling lobbying activities~\citep{HR89}, but since then, much analysis (especially that dealing with the Revenue Equivalence Theorem) has been done on game-theoretic auction theory.
When bidders have the same value distribution for the item,~\cite{MR03} showed that there is a symmetric equilibrium in auctions where the winner is the bidder with the highest bid. A significant analysis of all-pay auctions in full information settings was~\cite{ref1}, showing (aided by \cite{HR89}) the equilibrium states in various cases of all-pay auctions, and noting that most valuations (apart from the top two), are not relevant to the winner's strategies.

More recent work has extended the basic model. \cite{LPBR13} addressed issues of mergers and collusions, while several others directly addressed crowdsourcing models. \cite{DV09} detailed the issues stemming from needing to choose one auction from several, and \cite{CHS12} dealt with optimal mechanisms for crowdsourcing. 

The early major work on failures in auctions was \cite{MM87}, followed soon after by \cite{Mat87}, which introduced bidders who are not certain of how many bidders there will actually be at the auction. Their analysis showed that in first-price auctions (like our all-pay auction), risk averse bidders prefer to know the numbers, while it is the auctioneer's interest to hide that number. In the case of neutral bidders (such as ours), their model claimed that bidders were unaffected by the numerical knowledge. \cite{DKL89} claimed that experiments that allowed ``contingent'' bids (i.e., one submits several bids, depending on the number of actual participants) supported these results. \cite{MM00} presented a model where auction participants know the maximal number of bidders, but not how many will ultimately participate. However, the decision in their case was endogenous to the bidder, and therefore a reserve price has a significant effect in their model (though ultimately without change in expected revenue, in comparison to full-knowledge models). In contrast to that, our model, which assumes a little more information is available to the bidders (they know the maximal number of bidders, and the probability of failure), finds that in such a scenario, bidders are better off not having everyone show up, rather than knowing the real number of contestants appearing. Empirical work done on actual auctions~\citep{LY03} seems to support some of our theoretical findings (though not specifically in all-pay auction settings).

In our settings, the failure probabilities are public information and the failures are independent. Such failures have also been studied in other game-theoretic fields.
\cite{MTBK12} studied the effects of failures in congestion games, and showed that in some cases, the failures could be beneficial to the social welfare.
Some earlier work focused on agent redundancy and agent failures in cooperative games, studying various solution concepts in such games (see, e.g., ~\citep{BMFT11,BSS14}).

\section{Model}\label{sec:model}
We consider an all-pay auction with a single auctioned item that is commonly valued by all the participants. This is a restricted case of the model in \cite{ref1}, where players' item valuation could be different.

Formally, we assume that each of the $n$ bidders issues a bid of $b_{i}$, $i=1,\ldots,n$, and all bidders value the item at $1$. The highest bidders win the item and divide it among themselves, while the rest lose their bid. Thus, bidder $i$'s utility from a combination of bids $(b_{1},\ldots,b_{n})$ is given by:
\begin{equation}
\pi_{i}(b_{1},\ldots,b_{n})=
\begin{cases}
\frac{1}{\left|\arg\max\limits_{j} b_j\right|}-b_{i} & i \in\arg\max\limits_{j} b_{j}\\
-b_{i} & i\notin\arg\max\limits_{j} b_{j}.
\end{cases}
\end{equation}

We are interested in a symmetric equilibrium, which in this case, without possibility of failure, is unique~\citep{ref1,MR03}. It is a mixed equilibrium with full support of $[0,1]$, so that each bidder's bid is distributed in $[0,1]$ according to the same cumulative distribution function $F$, with the density function $f$ (since it is non-atomic, tie-breaking is not an issue). As we compare this case to that of no-failures, this is a case similar to that presented in \cite{ref1}, where various results on the behavior of non-cooperative bidders have been provided. We briefly give an overview of the results without failures in Subsection~\ref{noFail}.

When we allow bidders to fail, we assume that each of them has a probability of participating --- $p_{i}\in [0,1]$.
As a matter of convenience, we shall order the bidders according to their probabilities, so $0\leq p_{1}\leq p_{2}\leq\ldots\leq p_{n}\leq 1$.
If a bidder fails to participate, its utility is $0$.

\subsection{Auctions without Failures}\label{noFail}
The expected utility of any participant with a bid $b$ is:
\begin{equation}
\pi(b)=(1-b)\cdot Pr(winning \mid b) + (-b)\cdot Pr(losing \mid b)
\end{equation}
where $Pr(winning \mid b)$ and $Pr(losing \mid b)$ are the probabilities of winning or losing the item when bidding $b$, respectively. In a symmetric equilibrium with $n$ players, each of the bidders chooses their bid from a single bid distribution with a probability density function $f_{n}(x)$ and a cumulative distribution function $F_{n}(x)$. A player who bids $b$ can only win if all the other $n-1$ players bid at most $b$, which occurs with probability $F_n^{n-1}(b)$. Thus, the expected utility of a player bidding $b$ is given by:
\begin{equation}
\pi\left(b\right)=\left(1-b\right)F_{n}^{n-1}\left(b\right)-b\left(1-F_{n}^{n-1}\left(b\right)\right)=F_{n}^{n-1}\left(b\right)-b.
\end{equation}
The unique symmetric equilibrium is defined by the CDF $F_{n}(x)=x^{\frac{1}{n-1}}$~\citep{ref1}. 
This equilibrium has full support, and all points in the support yield the same expected utility to a player, 
 $\pi(0)=\pi(x)$ for all $x\in [0,1]$. Since $\pi(0)=0$, this means that for all bids, $\pi\left(b\right)=0$. 
The various properties of an auction without failures can be found in Table~\ref{nonFailAuctionTable}~\citep{LPBR13}.


\begin{table}
	\begin{center}
		\begin{tabular}{|l|c|}
			\hline
			\scriptsize{\bf Variable}& \scriptsize{\bf No Failures}\\
			\hline
			\scriptsize Expected bid & \scriptsize$\frac{1}{n}$\\
			$[$\scriptsize Variance$]$ & $\left[\frac{1}{2n-1}-\frac{1}{n^{2}}\right]$\\
			\hline
			\scriptsize Bidder utility & $0$\\
			$[$\scriptsize Variance$]$ & $\left[\frac{n-1}{n(2n-1)}\right]$\\
			\hline
			\scriptsize Sum-profit principal utility& $1$\\
			$[$\scriptsize Variance$]$& $\left[\frac{n}{2n-1}-\frac{1}{n}\right]$\\
			\hline
			\scriptsize Max-profit principal utility &$\frac{n}{2n-1}$ \\
			$[$\scriptsize Variance$]$& $\left[\frac{n(n-1)^{2}}{(3n-2)(2n-1)^{2}}\right]$\\
			\hline
		\end{tabular}\\
		\caption{The values, in expectation, of some of the variables when there is no possibility of failure}\label{nonFailAuctionTable}
	\end{center}
\end{table}

\section{Every Bidder with Own Failure Probability}\label{sec:ownF}
In this section, we assume that each bidder has its own probability for
participating in the auction, with $0\leq p_{1}\leq\ldots\leq p_{n}\leq1$.
We can assume without loss of generality, that each
bidder has a positive participating probability, that is, $p_{1}>0$.
If this is not the case, we can remove from the auction the bidders
with zero probability of participating.

\subsection{Equilibrium Properties}

Before we present a symmetric Nash equilibrium, we will characterize
any Nash equilibrium.
\begin{thm}
	In common values all-pay auction when the item value is $1$, if $p_{n-1}<1$ then there is a
	unique Nash equilibrium, in which the expected profit of every participating
	bidder is $\prod_{j=1}^{n-1}\left(1-p_{j}\right)$. Furthermore, there
	exists a continuous function $z:\left[0,1 - \prod_{j=1}^{n-1}\left(1-p_{j}\right) \right]\rightarrow\left[0,1\right]$,
	such that when a bidder $i$, has a positive density over an interval,
	they bid according to $F_{i}\left(x\right)=\frac{z\left(x\right)+p_{i}-1}{p_{i}}$
	over that interval, and if $p_{i}=p_{j}$ then $F_{i}=F_{j}$.
	\label{thm:positiveProfit}
\end{thm}
The proof of Theorem \ref{thm:positiveProfit} can be found at the
appendix.
As Theorem~\ref{thm:positiveProfit} applies to the case where $p_{n-1} < 1$, 
we now deal with the other case. 
\begin{thm}\label{thm:zeroProfit}
	In common values all-pay auction when the item value is $1$, if $p_{n-1}=1$ then in every Nash
	equilibrium the expected profit of every participating bidder is $0$.
	At least two bidders with $p=1$ randomize over $\left[0,1\right]$
	with each other player $i$ randomizing continuously over $\left(b_{i},1\right]$,
	$b_{i}>0$, and having an atomic point at $0$ of $\alpha_{i}\left(0\right)$.
	There exists a continuous function $z\left(x\right):\left[0,1\right]\rightarrow\left[0,1\right]$,
	such that when a bidder, $i$, has a positive density over an interval,
	they bid according to $F_{i}\left(x\right)=\frac{z\left(x\right)+p_{i}-1}{p_{i}}$
	over that interval. For every $i$, the atomic point at $0$
	is equals to $F_{i}\left(0\right)$. \label{thm;pn-1=00003D1}
\end{thm}
When there are at least two bidders with $p=1$, the auction approaches
the case without failures. The proof of Theorem~\ref{thm:zeroProfit} is 
a generalization of the case without agent failures~\citep{ref1}, and can be found at the appendix.

\subsection{Symmetric Equilibrium}

We are now ready to present a symmetric Nash equilibrium, we assume
that $0<p_{1}\leq\dots\leq p_{n}\leq1$. If $p_{n-1}<1$, from Theorem
\ref{thm:positiveProfit} it follows that the equilibrium is unique. If $p_{n-1}=1$
the equilibrium is not unique, except for two bidders with $p=1$,
every bidder can place an arbitrary atomic point at $0$. In the equilibrium
that we present, every bidder has an atomic point at $0$ of $0$, and
thus the equilibrium is symmetric.

In order to simplify the calculations, we add a ``dummy'' bidder,
with index 0, and $p_{0}=0$, adding a bidder that surely will not
participate in the auction, does not effect the other bidders and
therefore does not influence the equilibrium.

We begin by defining a few helpful functions. First, we define $\lambda=\prod\limits _{j=1}^{n-1}\left(1-p_{j}\right)$,
and we define the following expressions for all $1\leq k \leq n-1$:
\begin{equation}
H_{k}\left(x\right)=\left(\frac{\lambda+x}{\prod_{j=0}^{k-1}\left(1-p_{j}\right)}\right)^{\frac{1}{n-k}}
\end{equation}
and
\begin{equation}
\underline{s}_{k}=\left(1-p_{k}\right)^{n-k}\prod_{j=0}^{k-1}\left(1-p_{j}\right)-\lambda.
\end{equation}
For the virtual ``0'' index, we use $\underline{s}_{0}=1-\lambda$.
Note that because the $p_{i}$'s are ordered, so are the $\underline{s}_{i}$'s:
$1\geq\underline{s}_{0}\geq\underline{s}_{1}\geq\ldots\geq\underline{s}_{n-1}=0$.
\footnote{An equivalent definition of $\underline{s}_{k}$ is $\underline{s}_{k}=\left(1-p_{k}\right)^{n-k-1}\prod_{j=0}^{k}\left(1-p_{j}\right)-\lambda$,
	we alternate between those two definitions.
}

We are now ready to define the CDFs for our equilibrium, for every
bidder $1 \leq i \leq n-1$:

\begin{equation}
F_{i}\left(x\right)=\begin{cases}
1 & \text{\ensuremath{x\geq\underline{s}_{0}}}\\
\frac{H_{1}\left(x\right)+p_{i}-1}{p_{i}} & x\in\left[\underline{s}_{1},\underline{s}_{0}\right)\\
\vdots & \vdots\\
\frac{H_{k}\left(x\right)+p_{i}-1}{p_{i}} & x\in\left[\underline{s}_{k},\underline{s}_{k-1}\right)\\
\vdots & \vdots\\
\frac{H_{i}\left(x\right)+p_{i}-1}{p_{i}} & x\in\left[\underline{s}_{i},\underline{s}_{i-1}\right)\\
0 & x<\underline{s}_{i}.
\end{cases}
\label{eq:cdf_i}
\end{equation}
$F_{n}$, uniquely, while it is very similar to $F_{n-1}$ in its
piecewise composition, has an atomic point at $0$ of $1-\frac{p_{n-1}}{p_{n}}$,
so:
\begin{equation}
F_{n}\left(x\right)=\begin{cases}
1 & \text{\ensuremath{x\geq\underline{s}_{0}}}\\
\frac{H_{1}\left(x\right)+p_{n}-1}{p_{n}} & x\in\left[\underline{s}_{1},\underline{s}_{0}\right)\\
\vdots & \vdots\\
\frac{H_{k}\left(x\right)+p_{n}-1}{p_{n}} & x\in\left[\underline{s}_{k},\underline{s}_{k-1}\right)\\
\vdots & \vdots\\
\frac{H_{n-1}\left(x\right)+p_{n}-1}{p_{n}} & x\in\left(\underline{s}_{n-1},\underline{s}_{n-2}\right)\\
1-\frac{p_{n-1}}{p_{n}} & x=0\\
0 & x<0.
\end{cases}
\label{eq:cdf_n}
\end{equation}
Note that all CDFs are continuous and piecewise differentiable,\footnote{Note
	that when $\prod_{j=0}^{k-1}\left(1-p_{j}\right)=0$, and $H_{k}$
	is undefined for some $k$, then there is no range for which that
	$H_{k}$ is used.}
and when $p_{i}=p_{j}$ it follows that $F_{i}=F_{j}$; therefore,
this is a symmetric equilibrium. 
The intuition behind this equilibrium is that bidders that participate
rarely will usually bid high, while those that frequently participate
in auctions with less competition would more commonly bid low.
\begin{thm}
	The strategy profile 
	$F_1,\dots,F_n$
	defined in Equations~(\ref{eq:cdf_i}) and~(\ref{eq:cdf_n})
	 is an equilibrium, in which the expected profit of the bidders, if they haven't failed is $\lambda$.
	\label{thm:symEquil}
\end{thm}

\begin{proof}
	
	In the course of proving this is indeed
	a equilibrium, we shall calculate the expected utility of the bidders
	when they participate.

	When bidder $i$ bids according to this distribution, i.e., $x\in\left[\underline{s}_{k},\underline{s}_{k-1}\right)$
	for $1\leq k\leq i$:
	\begin{equation}
	\begin{array}{rl}
	\pi_{i}\left(x\right)= & \left(1-x\right)\prod\limits _{j=1;j\neq i}^{n}\left(p_{j}F_{j}\left(x\right)+1-p_{j}\right)-x\left(1-\prod\limits _{j=1;j\neq i}^{n}\left(p_{j}F_{j}\left(x\right)+1-p_{j}\right)\right)\\
	= & \prod\limits_{j=1;j\neq i}^{n}\left(p_{j}F_{j}\left(x\right)+1-p_{j}\right)-x\\
	= & \prod\limits_{j=1}^{k-1}\left(p_{j}F_{j}\left(x\right)+1-p_{j}\right)\prod\limits_{j=k;j\neq i}^{n}\left(p_{j}F_{j}\left(x\right)+1-p_{j}\right)-x\\
	= & \prod\limits_{j=1}^{k-1}\left(1-p_{j}\right)\prod\limits_{j=1;j\neq i}^{n}H_{k}\left(x\right)-x\\
	= & \prod\limits_{j=1}^{k-1}\left(1-p_{j}\right)H_{k}\left(x\right)^{n-k}-x\\
	= & \prod\limits_{j=1}^{k-1}\left(1-p_{j}\right)\left(\frac{\lambda+x}{\prod_{j=0}^{k-1}\left(1-p_{j}\right)}\right)-x\\
	= & \lambda.
	\end{array}
	\end{equation}	
	If bidder $i$ bids outside their support, i.e., $x\in$ for $i+1\leq k\leq n-1$,
	the same equation becomes:	
	\begin{equation}
	\begin{array}{rl}
	\pi_{i}(x)= & \prod\limits_{j=1;j\neq i}^{k-1}\left(1-p_{j}\right)\prod\limits_{j=k}^{n}H_{k}\left(x\right)-x\\
	= & \prod\limits_{j=1;j\neq i}^{k-1}\left(1-p_{j}\right)\left(\frac{\lambda+x}{\prod_{j=1}^{k-1}(1-p_{j})}\right)^{\frac{n-k+1}{n-k}}-x\\
	= & \frac{\lambda+x}{1-p_{i}}\left(\frac{\lambda+x}{\prod_{j=1}^{k-1}(1-p_{j})}\right)^{\frac{1}{n-k}}-x.
	\end{array}
	\end{equation}	
	Now,
	\begin{equation}
	\begin{array}{rl}
	x< & \underline{s}_{k-1}=\left(1-p_{k-1}\right)^{n-k+1}\prod_{j=1}^{k-2}\left(1-p_{j}\right)-\lambda\\
	= & \left(1-p_{k-1}\right)^{n-k}\prod_{j=1}^{k-1}\left(1-p_{j}\right)-\lambda
	\end{array}
	\end{equation}
	and hence $\lambda+x<(1-p_{k-1})^{n-k}\prod_{j=1}^{k-1}(1-p_{j})$. Plugging it all together, 	
	\begin{equation}
	\begin{array}{rl}
	\pi_{i}(x)< & \frac{\lambda+x}{1-p_{i}}\left(\frac{(1-p_{k-1})^{n-k}\prod_{j=1}^{k-1}(1-p_{j})}{\prod_{j=1}^{k-1}(1-p_{j})}\right)^{^{\frac{1}{n-k}}}-x\\
	= & \frac{\lambda+x}{1-p_{i}}\left(1-p_{k-1}\right)-x\\
	= & \left(\lambda+x\right)\frac{p_{i}-p_{k-1}}{1-p_{i}}+\lambda.
	\end{array}
	\end{equation}	
	Finally, as $i+1\leq k$, $p_{i}\leq p_{k-1}$, hence $p_{i}-p_{k-1}\leq 0$, and therefore $\pi_{i}(x)<\lambda$.	
\end{proof}
\subsection{Profits}

When a bidder actually participates their expected utility, in the equilibrium, is $\lambda$, and therefore the overall expected utility for bidder $i$ is $p_{i}\lambda$ (which, naturally, decreases with $n$). Notice that, as is to be expected, a bidder's profit rises the less reliable their fellow bidders are, or the fewer participants the auction has. However, the most reliable of the bidders does not affect the profits of the rest. If a bidder can set its own participation rate, if there is no bidder with $p_{j}=1$, that is the best strategy; otherwise, the optimal probability should be $\frac{1}{2}$, as that maximizes $p_{i}(1-p_{i})\prod_{j=1;j\neq i}^{n-1}(1-p_{j})$.

\subsubsection{Expected Bid}

In order to calculate the expected bid by each bidder, we need to calculate the bidders' equilibrium PDF, for $1\leq i\leq n-1$: 
\begin{equation}
f_{i}\left(x\right)=\begin{cases}
0 & \text{\ensuremath{x\geq\underline{s}_{0}}}\\
\frac{\left(\lambda+x\right)^{\frac{2-n}{n-1}}}{p_{i}\left(n-1\right)} & x\in\left[\underline{s}_{1},\underline{s}_{0}\right)\\
\vdots & \vdots\\
\frac{\left(\lambda+x\right)^{\frac{k+1-n}{n-k}}}{p_{i}\left(n-k\right)\prod_{j=0}^{k-1}\left(1-p_{j}\right)^{\frac{1}{n-k}}} & x\in\left[\underline{s}_{k},\underline{s}_{k-1}\right)\\
\vdots & \vdots\\
\frac{\left(\lambda+x\right)^{\frac{i+1-n}{n-i}}}{p_{i}\left(n-i\right)\prod_{j=0}^{i-1}\left(1-p_{j}\right)^{\frac{1}{n-i}}}\,\,\,\, & x\in\left[\underline{s}_{i},\underline{s}_{i-1}\right)\\
0 & x<\underline{s}_{i}
\end{cases}
\end{equation}
and $f_{n}\left(x\right)=\frac{p_{n-1}}{p_{n}}f_{n-1}\left(x\right)$.
In the equilibrium, the expected bid of bidder $i$, for $1 \leq i\leq n-1$ is:
\begin{equation}
\mathbb{E}\left[bid_{i}\right]=\sum_{k=1}^{i}\int\limits _{\underline{s}_{k}}^{\underline{s}_{k-1}}x f_{i}\left(x\right)\,\mathrm{d}x.
\end{equation}

\begin{thm}
	For every $1 \leq i \leq n-1$: 
	\begin{equation}
	\begin{split}\mathbb{E}\left[bid_{i}\right]= & \frac{1}{p_{i}}\left(\frac{1}{n}+\sum_{k=1}^{i}\frac{\left(1-p_{k}\right)^{n-k}\prod_{j=1}^{k}\left(1-p_{j}\right)}{\left(n-k\right)\left(n-k+1\right)}\right.\\
	& \qquad \left.-\frac{\left(1-p_{i}\right)^{n-i}\prod_{j=1}^{i}\left(1-p_{j}\right)}{n-i}-p_{i}\lambda\right)
	\end{split}
	\end{equation}
	and 
	\begin{equation}
	\mathbb{E}\left[bid_{n}\right]=\frac{p_{n-1}}{p_{n}}\mathbb{E}\left[bid_{n-1}\right].
	\end{equation}
	\label{thm:expectedBid}
\end{thm}
The expected bid decreases with $n$, indicating, as in the no-failure model, that as more bidders participate, the chance of losing increases, causing bidders to lower their exposure.
The proof of Theorem~\ref{thm:expectedBid} can be found at the
appendix.

\subsubsection{Auctioneer - Sum-Profit Model}
In the equilibrium,
the expected profit of the auctioneer in the sum-profit model is given
by:
\begin{equation}
\mathbb{E}\left[AP\right]=\sum\limits_{i=1}^{n}p_{i}\mathbb{E}\left[bid_{i}\right].
\end{equation}
When summing over all bidders, we receive a much simpler
expression.

\begin{thm}	
	The sum-profit auctioneer's equilibrium profits are:
\begin{equation}
	\sum_{i=1}^{n}p_{i}\mathbb{E}\left[bid_{i}\right]=1-\lambda\left(1+\sum_{i=1}^{n-1}p_{i}\right).
\end{equation}
	\label{thm:sumProfit}
\end{thm}
In this case, growth with $n$ is monotonic increasing, and hence, any addition to $n$ is a net positive for the sum-profit auctioneer. The proof can be found at the appendix. 

\subsubsection{Auctioneer - Max-Profit Model}

To calculate a max-profit auctioneer's profits, we need to first define the max-profit auctioneer's profits equilibrium CDF: 
\begin{equation}
G\left(x\right)=\prod_{i=1}^{n}\left(p_{i}F_{i}\left(x\right)+1-p_{i}\right)
\end{equation}
That is, 
\begin{equation}
G\left(x\right)=\begin{cases}
1 & x\geq\underline{s}_{0}\\
\left(\lambda+x\right)^{\frac{n}{n-1}} & x\in\left[\underline{s}_{1},\underline{s}_{0}\right)\\
\vdots & \vdots\\
\frac{\left(\lambda+x\right)^{\frac{n-k+1}{n-k}}}{\prod_{j=0}^{k-1}\left(1-p_{j}\right)^{\frac{1}{n-k}}}\,\,\, & x\in\left[\underline{s}_{k},\underline{s}_{k-1}\right)\\
\vdots & \vdots\\
\frac{\left(\lambda+x\right)^{2}}{\prod_{j=0}^{n-2}\left(1-p_{j}\right)} & x\in\left[\underline{s}_{n-1},\underline{s}_{n-2}\right)\\
0 & x<0.
\end{cases}
\end{equation}
This is differentiable, and hence we can find $g\left(x\right)=\frac{\partial}{\partial x}G\left(x\right)$ and the max-profit auctioneer's expected profit.
\begin{thm}
	In the equilibrium, the max-profit auctioneer's profits are:	
	\begin{equation}
	\begin{split}
	\mathbb{E}\left[AP\right] = & \int\limits _{\underline{s}_{n-1}}^{\underline{s}_{0}}x g\left(x\right)\,\mathrm{d}x\\
	= & \frac{n}{2n-1}-\lambda+ \sum_{k=1}^{n-1}\left(\frac{(1-p_{k})^{2n-2k-1}\prod_{j=1}^{k}(1-p_{j})^{2}}{4(n-k)^{2}-1}\right).
	\end{split}
	\end{equation}	
	\label{thm:maxProfit}
\end{thm}
From Theorem~\ref{thm:maxProfit} we can see that the max-profit auctioneer would prefer to minimize $\lambda$, have two
reliable bidders ($p_{n}=p_{n-1}=1$), and the other $n-2$ bidders
as unreliable as possible. The proof 
of Theorem~\ref{thm:maxProfit} can be found at the
appendix.

\begin{example}\label{exm:exm}	
	Consider how four bidders interact. Our bidders have participation probability of
	$p_{1}=\frac{1}{3}$, $p_{2}=\frac{1}{2}$, $p_{3}=\frac{3}{4}$
	and $p_{4}=1$. Let us look at each bidder's equilibrium CDFs:	
	\begin{equation}
	\begin{array}{cl}
	F_{1}(x)= & \begin{cases}
	1 & x\geq\frac{11}{12}\\
	3\left(\frac{1}{12}+x\right)^{\frac{1}{3}}-2\quad\enskip & x\in\left[\frac{23}{108},\frac{11}{12}\right)\\
	0 & x<\frac{23}{108}
	\end{cases}\\
	\\
	F_{2}(x)= & \begin{cases}
	1 & x\geq\frac{11}{12}\\
	2\left(\frac{1}{12}+x\right)^{\frac{1}{3}}-1 & x\in\left[\frac{23}{108},\frac{11}{12}\right)\\
	2\left(\frac{3\left(\frac{1}{12}+x\right)}{2}\right)^{\frac{1}{2}}-1\; & x\in\left[\frac{1}{12},\frac{23}{108}\right)\\
	0 & x<\frac{1}{12}
	\end{cases}\\
	\\
	F_{3}(x)= & \begin{cases}
	1 & x\geq\frac{11}{12}\\
	\frac{4}{3}\left(\frac{1}{12}+x\right)^{\frac{1}{3}}-\frac{1}{3} & x\in\left[\frac{23}{108},\frac{11}{12}\right)\\
	\frac{4}{3}\left(\frac{3\left(\frac{1}{12}+x\right)}{2}\right)^{\frac{1}{2}}-\frac{1}{3} & x\in\left[\frac{1}{12},\frac{23}{108}\right)\\
	4\left(\frac{1}{12}+x\right)-\frac{1}{3} & x\in\left[0,\frac{1}{12}\right)\\
	0 & x<0
	\end{cases}\\
	\\
	F_{4}(x)= & \begin{cases}
	1 & x\geq\frac{11}{12}\\
	\left(\frac{1}{12}+x\right)^{\frac{1}{3}} & x\in\left[\frac{23}{108},\frac{11}{12}\right)\\
	\left(\frac{3\left(\frac{1}{12}+x\right)}{2}\right)^{\frac{1}{2}}\qquad\enskip & x\in\left[\frac{1}{12},\frac{23}{108}\right)\\
	3\left(\frac{1}{12}+x\right) & x\in\left(0,\frac{1}{12}\right)\\
	\frac{1}{4} & x=0\\
	0 & x<0
	\end{cases}
	\end{array}
	\end{equation}

A graphical illustration of the bidders' CDFs and PDFs can be found in the appendix. 
The expected utility for bidder 1 is $0.027$, for expected bid of $0.518$; for bidder 2, $0.041$ for expected bid of $0.394$; for bidder 3, $0.0625$ for expected bid of $0.277$; and for the last bidder, $0.083$ for expected bid of $0.207$.

A sum-profit auctioneer will see an expected profit of $0.0784$, while a max-profit one will get, in expectation, $0.490$.

As a comparison, in the case where we do not allow failures, the CDF of the bidders is $x^\frac{1}{3}$ with expected bid of $\frac{1}{4}$ and expected utility of $0$. The expected profit of the sum-profit auctioneer is $1$, while the expected profit of the max-profit auctioneer is $\frac{4}{7}$. 
\end{example}

\section{False Identity and Sabotage}\label{sec:sbtg}
\begin{center}	
	\begin{algorithm}
		\caption{Optimal Bid}\label{alg:opt}
		\renewcommand{\algorithmicrequire}{\textbf{Input:}}
		\renewcommand{\algorithmicensure}{\textbf{Output:}}	
		\begin{algorithmic}[1]
			\Require $i$, $r$, $p_{1},\dots,p_{n}$ and $p'_{r}$
			\Ensure The optimal bid for bidder $i$
			
			\State let $\lambda=\prod_{j=1}^{n-1}\left(1-p_{j}\right)$
			\For {$k=1, \dots ,\min\left\{ i,r\right\} $}
			\If {$\frac{1}{n-k}\in\left[p_{k-1},p_{k}\right]$}
			\State let $x_{i,k}=\left(1-\frac{1}{n-k}\right)^{n-k}\prod_{j=1}^{k-1}\left(1-p_{j}\right)-\lambda$
			\State let $\pi_{i,k}=\frac{1}{n-k}\left(1-\frac{1}{n-k}\right)^{n-k-1}\prod_{j=1}^{k-1}\left(1-p_{j}\right)\frac{p_{r}-p'_{r}}{p_{r}}+\lambda$
			\ElsIf {$\frac{1}{n-k} < p_{k-1}$}
			\State let $x_{i,k}=\underline{s}_{k-1}$
			\State let $\pi_{i,k}=p_{k-1}\left(1-p_{k-1}\right)^{n-k-1}\prod_{j=1}^{k-1}\left(1-p_{j}\right)\frac{p_{r}-p'_{r}}{p_{r}}+\lambda$
			\ElsIf {$\frac{1}{n-k}>p_{k}$}
			\State let $x_{i,k}=\underline{s}_{k}$
			\State let $\pi_{i,k}=p_{k}\left(1-p_{k}\right)^{n-k-1}\prod_{j=1}^{k-1}\left(1-p_{j}\right)\frac{p_{r}-p'_{r}}{p_{r}}+\lambda$		
			\EndIf
			\EndFor
			\State let $j\in\arg\max_{k}\pi_{i,k}$ \\
			\Return $x_{i,k}$
		\end{algorithmic}
	\end{algorithm}
\end{center}

Now, suppose our bidder can influence others' perceptions, and create a false sense of its participation probability.
What would its best strategy be, and how should the participation probability be altered? Any bid beyond $1 - \lambda$ is sure to win, but as that would give profit of less than $\lambda$, which is less than the expected profit for non-manipulators, it is not worthwhile. Therefore, our bidder will bid in its support, with the expected profit being $\lambda$.
However, 
 our bidder may increase its expected profit
 by trying to portray its participation probability as being as low as possible, thus lulling the other bidders with a false sense of security. Of course, this reduces the payment to auctioneers of any type, and therefore, they would try to expose such manipulation.

More interesting is the possibility of a player's changing another player's participation probability by using sabotage; thus our bidder would be the only bidder knowing the real participation probability. Our bidder, $i$, sabotages bidder $r$, with a perceived participation probability of $p_{r}$, changing its real participation probability to $p'_{r}$. Bidder $i$'s expected profit with bid $x$ is:
\begin{equation}
\pi_{i}(x)=\left(p'_{r}F_{r}(x)+1-p'_{r}\right)\prod_{j=1;j\neq i,r}^{n}\left(p_{j}F_{j}(x)+1-p_{j}\right).
\end{equation}
The values of this function change according to the relation between $r$, $i$ and $x$. To find the optimal strategy for a bidder, we must examine all the options.

\begin{thm}
	Let $p_1,\dots, p_n$ be the announced participation probabilities, and let $p'_r < p_r$ be bidder $r$ real participation probability. 
	For every $i \neq r$, Algorithm~\ref{alg:opt} finds the optimal bid for bidder $i$.
	\label{thm:sbg}	
\end{thm}

To summarize, bidder $i$ best interest is to bid in the intersection of its support and bidder $r$'s support. 
Given the index of the saboteur bidder, the index of
the sabotaged bidder, $p_{1},\dots,p_{n}$ and the participation probability
after the sabotage, Algorithm~\ref{alg:opt} finds the optimal bid. The full proof can be found at the appendix.

\section{Uniform Failure Probabilities}\label{sec:uni}	
	If we allow our bidders to have the same probability of failure (e.g.,
	when failures stem from weather conditions), many of the calculations
	become more tractable, and we are able to further understand the scenario.
	\subsection{Bids}
	
	As this case is a particular instance of the general case presented
	above, we can calculate the expected equilibrium bid of every bidder and its variance.
	\begin{thm}			
			The expected equilibrium bid of every bidder is:
			\begin{equation}
			\mathbb{E}\left[bid\right] = \frac{1}{n p}\left(1-\lambda\left(1+p\left(n-1\right)\right)\right)
		\end{equation}
		and the variance of the bid is:
			\begin{equation}
			\mathrm{Var}\left[bid\right]= \frac{1-(1-p)^{2n-1}}{(2n-1)p}-\frac{\left(1-(1-p)^{n}\right)^{2}}{n^{2}p^{2}}.
			\end{equation}
				The expected bid and the variance are 
			neither monotonic in $n$ nor in $p$.
	\label{prop:expbid}		
	\end{thm}

	\subsection{Profits }
	
	We are now ready to examine the profits of all the parties, the bidder
	and the auctioneer, both in the sum-profit model and the max-profit
	model.

	\subsubsection{Bidder}
	
	From the general case we may deduce that 
	expected equilibrium profit of every bidder is $p\left(1-p\right)^{n-1}$. 
	Note the profit decreases as $n$ increases, and is maximized when $p=\frac{1}{n}$.
	We can now compute the variance of bidder profit.
	\begin{thm}
			The variance of the bidder equilibrium profit is:
			\begin{equation}
			\mathrm{Var}\left[BP\right]= \frac{n-1}{n\left(2n-1\right)}-\frac{\left(1-p\right)^{n}}{n}+\left(p+\frac{1}{2n-1}\right)\left(1-p\right)^{2n-1}.
			\end{equation}			
			And the variance is monotonic increasing in $p$.
	\label{prop:varbid}\end{thm}

	\subsubsection{Auctioneer - Sum-Profit Model}
	
	The expected bid of every bidder is $\frac{1}{pn}\left(1-\left(1-p\right)^{n-1}\left(1+p\left(n-1\right)\right)\right)$,
	therefore, the expected profit of the sum-profit auctioneer, in the equilibrium, is: 
	\begin{equation}
	\begin{array}{rl}
	\mathbb{E}\left[AP\right]= & np\cdot\mathbb{E}\left[bid\right]\\
	= & 1-\left(1-p\right)^{n-1}\left(1+p\left(n-1\right)\right)
	\end{array}
	\end{equation}
	which increases with $p$ and $n$. Therefore, the auctioneer best interest is to have as many bidders as possible.
	Note that as $n$ grows, the auctioneer's expected revenue approaches
	that of the no-failure case.
	From Theorem~\ref{prop:varbid} we get the variance of the auctioneer equilibrium profit in the sum profit model:
		\begin{equation}
		\begin{array}{rl}
		\mathrm{Var}\left[AP\right]= & n p^{2}\mathrm{Var}\left[bid\right]\\
		= & \frac{n\cdot p\left(1-(1-p)^{2n-1}\right)}{2n-1}-\frac{\left(1-(1-p)^{n}\right)^{2}}{n}.
		\end{array}
		\end{equation}

	\subsubsection{Auctioneer - Max-Profit Model}
	
	For the max-profit auctioneer, the expected profit in equilibrium is:
	\begin{equation}
	\frac{n}{2n-1}+\frac{n-1}{2n-1}\left(1-p\right)^{2n-1}-\left(1-p\right)^{n-1}
	\end{equation}
	which is, monotonically increasing in $p$ and $n$ (for $n \geq 1$);
	for large enough $n$ it approaches the expected revenue in the no-failure
	case. 
		\begin{thm}
		The variance of the auctioneer equilibrium profit in the max profit model is:
			\begin{equation}
			\begin{array}{rl}
			\mathrm{Var}\left[AP\right]= & \left(1-p\right)^{2n-2}-\frac{2n\left(1-p\right)^{n-1}}{2n-1}+\frac{n}{3n-2}-\frac{2\left(n-1\right)^{2}\left(1-p\right)^{3n-2}}{\left(3n-2\right)\left(2n-1\right)}\\
			& -\left(\frac{n}{2n-1}+\frac{n-1}{2n-1}\left(1-p\right)^{2n-1}-\left(1-p\right)^{n-1}\right)^{2}.
			\end{array}
			\end{equation}\label{prop:varmax}
		\end{thm}	
	
	\section{Conclusion and Discussion }
	
	Bidders failing to participate in auctions happen commonly, as people
	choose to apply to one job but not another, or to participate in the
	Netflix challenge but not a similar challenge offered by a competitor.
	Examining these scenarios enables us to understand certain fundamental
	issues in all-pay auctions. In the complete reliability, classic versions,
	each bidder has an expected revenue of $0$.
	In contrast, in a limited reliability
	scenario, such as the one we dealt with, bidders have positive expected
	revenue, and are incentivized to participate in the auction.
	Auctioneers, on the other hand, mostly lose their strong control of
	the auction, and no longer pocket almost all revenues involved in
	the auction. However, by influencing participation probabilities,
	max-profit auctioneers can effectively increase their revenue in comparison
	to the no-failure model.

The idea of the equilibrium we explored was that frequent participants could allow themselves to bid lower, as there would be
plenty of contests where they would be one of the few participants,
and hence win with smaller bids. Infrequent bidders, on the other
hand, would wish to maximize the few times they participate,
and therefore bid fairly high bids. As exists in the no-failure case as well,
as more and more participants join, there is a concentration of bids
at lower price points, as bidders are more afraid of the fierce competition.
Hence, it is fairly easy to see in all of our results that as $n$
approached larger numbers, the various variables were closer and closer
to their no-failure brethren.

There is still much left to explore in these models --- not only more
techniques of manipulation by bidders and potential incentives by
auctioneers, but also further enrichment of the model. Currently,
participation rates are not influenced by other bidders' probability
of participation, but, obviously, many scenarios in real-life have,
effectively, a feedback loop in this regard.
We assumed that the item is commonly valued by all the bidders and the cost of
effort is common, which is not always the case.
A future research could examine a more realistic model with heterogeneous costs or valuations.
In our model the failure happened before the bidder placed their bid,
but in other models
the failure could happen after the bidders place their bid and before the auctioneer collected the bid. 
Finding a suitable model
for such interactions, while an ambitious goal, might help us gain
even further insight into these types of interactions.

\bibliographystyle{plainnat}
\bibliography{main}

\begin{thebibliography}{14}
\providecommand{\natexlab}[1]{#1}
\providecommand{\url}[1]{\texttt{#1}}
\expandafter\ifx\csname urlstyle\endcsname\relax
  \providecommand{\doi}[1]{doi: #1}\else
  \providecommand{\doi}{doi: \begingroup \urlstyle{rm}\Url}\fi

\bibitem[Bachrach et~al.(2011)Bachrach, Meir, Feldman, and Tennenholtz]{BMFT11}
Yoram Bachrach, Reshef Meir, Michal Feldman, and Moshe Tennenholtz.
\newblock Solving cooperative reliability games.
\newblock In \emph{Proceedings of the 28th Conference in Uncertainty in
  Artificial Intelligence (UAI)}, pages 27--34, Barcelona, Spain, July 2011.

\bibitem[Bachrach et~al.(2014)Bachrach, Savani, and Shah]{BSS14}
Yoram Bachrach, Rahul Savani, and Nisarg Shah.
\newblock Cooperative max games and agent failures.
\newblock In \emph{Proceedings of the 13th International Joint Conference on
  Autonomous Agents and Multiagent Systems (AAMAS)}, pages 29--36, Paris,
  France, May 2014.

\bibitem[Baye et~al.(1996)Baye, Kovenock, and Vries]{ref1}
Michael~R. Baye, Dan Kovenock, and Casper~G. Vries.
\newblock The all-pay auction with complete information.
\newblock \emph{Economic Theory}, 8\penalty0 (2):\penalty0 291--305, 1996.
\newblock ISSN 1432-0479.

\bibitem[Chawla et~al.(2012)Chawla, Hartline, and Sivan]{CHS12}
Shuchi Chawla, Jason~D. Hartline, and Balasubramanian Sivan.
\newblock Optimal crowdsourcing contests.
\newblock In \emph{Proceedings of the 23rd Annual ACM-SIAM Symposium on
  Discrete Algorithms (SODA)}, pages 856--868, Kyoto, Japan, 2012. SIAM.

\bibitem[DiPalantino and Vojnovi{\'c}(2009)]{DV09}
Dominic DiPalantino and Milan Vojnovi{\'c}.
\newblock Crowdsourcing and all-pay auctions.
\newblock In \emph{Proceedings of the 10th ACM conference on Electronic
  commerce}, pages 119--128, Stanford, California, July 2009.

\bibitem[Dyer et~al.(1989)Dyer, Kagel, and Levin]{DKL89}
Douglas Dyer, John~H. Kagel, and Dan Levin.
\newblock Resolving uncertainty about the number of bidders in independent
  private-value auctions: An experimental analysis.
\newblock \emph{The RAND Journal of Economics}, 20\penalty0 (2):\penalty0
  268--279, 1989.

\bibitem[Hillman and Riley(1989)]{HR89}
Arye~L. Hillman and John~G. Riley.
\newblock Politically contestable rents and transfers.
\newblock \emph{Economics \& Politics}, 1\penalty0 (1):\penalty0 17--39, March
  1989.

\bibitem[Lev et~al.(2013)Lev, Polukarov, Bachrach, and Rosenschein]{LPBR13}
Omer Lev, Maria Polukarov, Yoram Bachrach, and Jeffrey~S. Rosenschein.
\newblock Mergers and collusion in all-pay auctions and crowdsourcing contests.
\newblock In \emph{Proceedings of the 12th International Coference on
  Autonomous Agents and Multiagent Systems (AAMAS)}, pages 675--682, St. Paul,
  Minnesota, May 2013.

\bibitem[Lu and Yang(2003)]{LY03}
Dennis Lu and Jing Yang.
\newblock Auction participation and market uncertainty: Evidence from canadian
  treasury auctions.
\newblock In \emph{Conference paper presented at the Canadian Economics
  Association 37th Annual Meeting}, 2003.

\bibitem[Maskin and Riley(2003)]{MR03}
Eric Maskin and John Riley.
\newblock Uniqueness of equilibrium in sealed high-bid auctions.
\newblock \emph{Games and Economic Behavior}, 45\penalty0 (2):\penalty0
  395--409, November 2003.

\bibitem[Matthews(1987)]{Mat87}
Steven Matthews.
\newblock Comparing auctions for risk averse buyers: A buyer's point of view.
\newblock \emph{Econometrica}, 55\penalty0 (3):\penalty0 633--646, May 1987.

\bibitem[McAfee and McMillan(1987)]{MM87}
Randolph~Preston McAfee and John McMillan.
\newblock Auctions with a stochastic number of bidders.
\newblock \emph{Journal of Economic Theory}, 43\penalty0 (1):\penalty0 1--19,
  October 1987.

\bibitem[Meir et~al.(2012)Meir, Tennenholtz, Bachrach, and Key]{MTBK12}
Reshef Meir, Moshe Tennenholtz, Yoram Bachrach, and Peter Key.
\newblock Congestion games with agent failures.
\newblock In \emph{Proceedings of the 26th National Conference on Artificial
  Intelligence (AAAI)}, pages 1401--1407, Toronto, Canada, July 2012.

\bibitem[Menezes and Monteiro(2000)]{MM00}
Flavio~M. Menezes and Paulo~K. Monteiro.
\newblock Auctions with endogenous participation.
\newblock \emph{Review of Economic Design}, 5\penalty0 (1):\penalty0 71--89,
  March 2000.

\end{thebibliography}

\newpage{}
\appendix

\subsection*{Proof of Theorem \ref{thm:positiveProfit} \label{sec:app1}}
Before we prove Theorem \ref{thm:positiveProfit}, we begin with some
notations; let $0<p_{1}\leq\dots\leq p_{n} \leq1$ be the participation probability of the bidders. we assume
that $p_{n-1}<1$, that is, at most one bidder surely participates
in the auction. Let $F_{1},\dots F_{n}$, be a Nash equilibrium, such
that $F_{i}$ denotes bidder $i$ bid distribution, where $F_{i}\left(x\right)=\Pr\left[bid_{i}\leq x\right]$,
note that $F_{i}$ is a right-continuous function, and $F_{i}$ has
a left-discontinuous point at $x$ if and only if the bidder has an atomic point at
$x$.

First, we should note that there indeed exists an equilibrium, as 
the strategy profile 
$F_1,\dots,F_n$ defined in Equations~(\ref{eq:cdf_i}) and~(\ref{eq:cdf_n}) is an equilibrium (see Theorem~\ref{thm:symEquil}).

All the bidders value the item at $1$, hence every bidder has no
incentive to bid more than $1$, bids less the $0$ are ruled out,
that is, the bids are drawn from $\left[0,1\right].$ 

For every $i$ let 
\begin{equation}
\alpha_{i}\left(x\right)=F_{i}\left(x\right)-\lim_{y\rightarrow x^{-}}F_{i}\left(y\right),
\end{equation}
that is, $\alpha_{i}\left(x\right)>0$ if and only if bidder $i$ has an atomic
point at $x$.

For every $i$ and $x$ denote by $\Pr\left(i\,\text{wins}\mid x\right)$
the probability that bidder $i$ wins the item, when they might share
the item with other bidders,\footnote{
For example, if when bidding $x$, bidder $i$ has a probability of
$p_{1}$ to be the only winner of the item, probability of $p_{2}$
to share the item with one other bidder and probability of $0$ to
share the item with other 2 or more bidders, then $\Pr\left(i\, \text{wins} \mid x\right)=p_{1}+\frac{1}{2}p_{2}$.}
 if all the other bidders are bidding according to $F_{-i}$;
 and denote by $\pi_i\left(x\right)$ the expected profit of bidder $i$ from bidding $x$.
 That is, $\Pr\left(i\,\text{wins}\mid x\right)=\pi_{i}\left(x\right)+x$.

Let $\underline{s}_{i}=\inf_{x}\left\{x\in \left[0,1\right]: F_{i}\left(x\right)>F_{i}\left(0\right)\right\} $
and $\overline{s}_{i}=\inf_{x}\left\{x \in \left[0,1\right]: F_{i}\left(x\right)=1\right\} $,
the lower and upper bound of player $i$'s equilibrium bid distribution,
excluding the atomic point at $0$ (if exists), respectively. Let $\pi_{i}^{*}$
be player $i$'s equilibrium profit, if they have not failed. We said
that $F_{i}$ is strictly increasing at $x>0$, if for every $\varepsilon>0$:
$F_{i}\left(x\right)>F_{i}\left(x-\varepsilon\right)$.
If $F_{i}$ is strictly increasing at $x$, then $x$ is in bidder $i$'s support.

The following lemmata characterize any Nash equilibrium, $F_1,\dots,F_n$ when $p_{n-1}<1$,
and help us to show that the equilibrium is indeed unique.
\begin{lem}
\label{lem:atom}For every $i$ and for every $x>0$: $\alpha_{i}\left(x\right)=0$.\end{lem}
\begin{proof}
Falsely assume there is a bidder, say $i$, with $\alpha_{i}\left(x\right)>0$
for $x>0$.
We shall consider three cases:
\begin{itemize}
	\item There is other bidder, say $j$, 	
	 such that $\alpha_j\left(x\right) > 0$. 
 	Bidder $j$ can bid $0$, with an expected
 	profit of at least $0$. Since $x>0$, it must hold that $\Pr\left(j\,\text{wins}\mid x\right)>0$. 
 	In addition, $x<1$ as $\Pr\left(j\,\text{wins}\mid x\right)<1$.
 	$\Pr\left(j\,\text{wins}\mid \cdot \right)$
 	 has an upward jump at $x$ (since a tie with $i$ is no longer possible), therefore there is 
 	 a sufficiently small $\varepsilon$ such that 
 	 $\Pr\left(j\,\text{wins}\mid x+\varepsilon\right)-\Pr\left(j\,\text{wins}\mid x\right) > \varepsilon$.
 	 Now,
 	\begin{equation}
 	\pi_j(x) = \Pr\left(j \text{wins}\mid x \right) - x < \Pr\left(j\,\text{wins}\mid x + \varepsilon \right) - \left(x + \varepsilon\right) = \pi_j(x + \varepsilon),
 	\end{equation}
 and thus
	 bidder $j$ has an incentive to raise $x$ and place an atom at $ x + \varepsilon$. 
	\item 
	 There is no other bidder with a mass at $x$, but there is other bidder, 
	 say $j$, 	 such that $x$ is in bidder $j$'s support. 
	 As 	$\Pr\left(j\,\text{wins}\mid \cdot \right)$
	 has an upward jump at $x$,
	it pays for $j$ to transfer mass from an 
	$\varepsilon$-neighborhood below $x$ 
	to some $\delta$-neighborhood above $x$.
	\item	
	There is no other bidder with $x$ in its support, in this case bidder
	$i$ has an incentive to reduce $x$ by a sufficiently small $\varepsilon$.
\end{itemize}
 Hence, for every $x>0$ there is no bidder with an atomic point at $x$. 
\end{proof}

\begin{lem}
	There is at most one bidder with an atomic point at $0$.\label{lem:pn-1}\end{lem}
\begin{proof}
	Suppose there are two bidders,
	 $i$ and $j$,
	 with an atomic point at $0$.
	 	$\Pr\left(j\,\text{wins}\mid \cdot \right)$
	 	has an upward jump at $0$, therefore,
	for a sufficiently small $\varepsilon$: $\Pr\left(j\,\text{wins}\mid\varepsilon\right) - \Pr\left(j\,\text{wins}\mid0\right) > \varepsilon$,
	and therefore 
	\begin{equation}
	\pi_j(0) = \Pr\left(j\,\text{wins}\mid0\right) < \Pr\left(j\,\text{wins}\mid\varepsilon\right) - \varepsilon = \pi_j(\varepsilon).
	\end{equation}
	Hence bidder $j$ has an incentive to raise its bid. 
\end{proof}

The bids are distributed with a non-atomic distribution (except maybe at $0$), therefore we do not need to address cases of ties between them, and we have that
for every $i$ and for every $x>0$, the following holds: 
\begin{equation}
\Pr\left(i\,\text{wins}\mid x\right)=\prod_{j\neq i}\left(p_{j}F_{j}\left(x\right)+1-p_{j}\right).
\end{equation}

\begin{lem}
For $x,y>0$ and $i$, if $x$ and $y$ are in the support of bidders $i$,
then $\pi_i(x) = \pi_i(y)$. \label{lem:supportEQ}
\end{lem}

\begin{proof}
	Otherwise, assume without loss of generality that $\pi_i(x) > \pi_i(y)$. As the bids are distributed with a non-atomic distribution, (except maybe one bidder with an atomic at $0$), the utility function $\pi_i$ is continuous in $\left(0,1\right)$. Since $\pi_i(x) > \pi_i(y)$, 
	it pays for $j$ to transfer mass from an 
	$\varepsilon$-neighborhood around $y$ 
	to some $\delta$-neighborhood around $x$.
\end{proof}

\begin{lem}
	For every $i$ and $x>0$, if $x$ is in the support of bidder $i$, then $\pi_i(x) = \pi^*_i$. 
	Moreover, if $\alpha_i(0)>0$, then $\pi_i(0) = \pi^*_i$
	 \label{lem:supportEq2}
\end{lem}
\begin{proof}
	If bidder $i$ does not have an atomic point at $0$ then the claim follows immediately 
from Lemma~\ref{lem:atom} and Lemma~\ref{lem:supportEQ}. If bidder $i$ has an atomic point at $0$ then it is sufficient to show that for some $x$ in the support of bidder $i$: $\pi_i(0) = \pi_i(x)$.
If $\pi_i(0) > \pi_i(x)$ then 	it pays for $i$ to transfer mass from an 
$\varepsilon$-neighborhood of $x$ 
to $0$;
 and if
 $\pi_i(0) < \pi_i(x)$ then 	it pays for $i$ to transfer mass from
 $0$ to 
 an 
 $\varepsilon$-neighborhood of $x$. 
 \end{proof}

 \begin{lem}
For every $i$ and $j$, if $\pi_{i}^{*}<\pi_{j}^{*}$ then $\overline{s}_{j}>0$.\label{lem:upper}\end{lem}
\begin{proof}
Suppose that $\pi_{i}^{*}<\pi_{j}^{*}$ and $\overline{s}_{j}=0$.
Note that $0 \leq \pi_{i}^{*}<\pi_{j}^{*}$.
If $\overline{s}_{j}=0$
then bidder $j$ has an atomic point at $0$ of $1$, and $\pi_{j}^{*}=\pi_j\left(0\right)=\Pr\left(j\, \text{wins} \mid 0\right)$.
$\Pr\left(i\,\text{wins}\mid \cdot \right)$
has an upward jump at $0$, therefore
for a small $\varepsilon>0$, bidder $i$ can bid $\varepsilon$ and
outbid bidder $j$, and hence $\pi_i\left(\varepsilon\right) \geq \pi^*_j$.
As $\pi^*_i \geq \pi_i\left(\varepsilon\right)$,
we have that $\pi_{i}^{*}\geq\pi_{j}^{*}$
in contradiction to the assumption,
thus $\overline{s}_{j}>0$. \end{proof}
\begin{lem}
For every $i$ and $j$, $\pi_{i}^{*}=\pi_{j}^{*}$\label{lem:lambda}.\end{lem}
\begin{proof}
Suppose there exist $i$ and $j$ such that $\pi_{i}^{*}\neq\pi_{j}^{*}$,
and without loss of generality $\pi_{i}^{*}<\pi_{j}^{*}$. From Lemma
\ref{lem:upper}, it must hold that $\overline{s}_{j}>0$. Since $\overline{s}_{j}>0$
from Lemma \ref{lem:atom}, the expected profit of bidder $j$ from
bidding $\overline{s}_{j}$ is:
\begin{equation}
\pi_{j}\left(\overline{s}_{j}\right)=\left(p_{i}F_{i}\left(\overline{s}_{j}\right)+1-p_{i}\right)\prod_{k\neq i,j}\left(p_{k}F_{k}\left(\overline{s}_{j}\right)+1-p_{k}\right)-\overline{s}_{j}.
\end{equation}
And from Lemma~\ref{lem:supportEq2}, as 
$\overline{s}_{j} > 0$ and in bidder $j$'s support, $\pi_{j}\left(\overline{s}_{j}\right)=\pi_{j}^{*}$.
Bidder $i$'s expected profit from bidding $\overline{s}_{j}$ is: 
\begin{equation}
\pi_{i}\left(\overline{s}_{j}\right)=\left(p_{j}F_{j}\left(\overline{s}_{j}\right)+1-p_{j}\right)\prod_{k\neq i,j}\left(p_{k}F_{k}\left(\overline{s}_{j}\right)+1-p_{k}\right)-\overline{s}_{j}
\end{equation}
 $\overline{s}_{j}$ might not be in bidder $i$'s support, hence
$\pi_{i}^{*}\geq\pi_{i}\left(\overline{s}_{j}\right)$, $F_{j}\left(\overline{s}_{j}\right)=1$
and $F_{i}\left(\overline{s}_{j}\right)\leq1$, thus 
\begin{equation}
\begin{array}{rl}
\pi_{i}^{*}\geq & \pi_{i}\left(\overline{s}_{j}\right)\\
= & \left(p_{j}F_{j}\left(\overline{s}_{j}\right)+1-p_{j}\right)\prod_{k\neq i,j}\left(p_{k}F_{k}\left(\overline{s}_{j}\right)+1-p_{k}\right)-\overline{s}_{j}\\
= & \prod_{k\neq i,j}\left(p_{k}F_{k}\left(\overline{s}_{j}\right)+1-p_{k}\right)-\overline{s}_{j}\\
\geq & \left(p_{i}F_{i}\left(\overline{s}_{j}\right)+1-p_{i}\right)\prod_{k\neq i,j}\left(p_{k}F_{k}\left(\overline{s}_{j}\right)+1-p_{k}\right)-\overline{s}_{j}\\
= & \pi_{j}\left(\overline{s}_{j}\right)\\
= & \pi_{j}^{*}
\end{array}
\end{equation}
This contradicts the assumption that $\pi_{i}^{*}<\pi_{j}^{*}$, hence
for every $i$ and $j$: $\pi_{i}^{*}=\pi_{j}^{*}$.
\end{proof}
Let $\lambda$ denote the expected profit of every bidder if have not failed, that is the expected utility of bidder $i$ is $p_{i}\lambda$.
Bidder $n$ can bid $0$, and win the item if all the other bidders
failed to participate in the auction, therefore $\lambda\geq\prod_{i=1}^{n-1}\left(1-p_{i}\right)>0$. 
\begin{lem}
If $\underline{s}_{i}<\underline{s}_{j}$ for some $i\neq j$, then
$p_{i}>p_{j}$.\label{lem:sisj}\end{lem}
\begin{proof}
Assume to the contrary that there exist $i$ and $j$ such that $\underline{s}_{i}<\underline{s}_{j}$
but $p_{i}\leq p_{j}$. 
\begin{itemize}
\item If $\alpha_{i}\left(0\right)>0$, from Lemma \ref{lem:pn-1} there is no other bidder with $\alpha\left(0\right)=0$ and from Lemma~\ref{lem:supportEq2}: $\lambda=\pi_{i}\left(0\right)=\prod_{k\neq i}\left(1-p_{k}\right)$.
Since bidder $i$ has an atomic point at $0$, it holds that 
\begin{equation}
\begin{array}{rl}
\pi_{j}\left(0\right)= & \left(1-p_{i}+\frac{1}{2}\left(p_{i}\cdot\alpha_{i}\left(0\right)\right)\right) \prod_{k\neq i,j}\left(1-p_{k}\right)\\
> & \left(1-p_{i}\right) \prod_{k\neq i,j}\left(1-p_{k}\right)\\
\geq & \left(1-p_{j}\right) \prod_{k\neq i,j}\left(1-p_{k}\right)\\
= & \lambda
\end{array}
\end{equation}

\item If bidder $i$ does not have an atomic point at $0$, let 
$x\in\left(\underline{s}_{i},\underline{s}_{j}\right)$, since $x>0$ it holds that 
\begin{equation}
\begin{array}{rl}
\pi_{i}\left(x\right)= & \left(p_{j}F_{j}\left(x\right)+1-p_{j}\right)\prod_{k\neq i,j}\left(p_{k}F_{k}\left(x\right)+1-p_{k}\right)-x\\
= & \left(1-p_{i}\right)\prod_{k\neq i,j}\left(p_{k}F_{k}\left(x\right)+1-p_{k}\right)-x\\
= & \lambda
\end{array}
\end{equation}
Bidder $j$'s expected profit from bidding $x$ is: 
\begin{equation}
\begin{array}{rl}
\pi_{j}\left(x\right)= & \left(p_{i}F_{i}\left(x\right)+1-p_{i}\right)\prod_{k\neq i,j}\left(p_{k}F_{k}\left(x\right)+1-p_{k}\right)-x\\
> & \left(1-p_{i}\right)\prod_{k\neq i,j}\left(p_{k}F_{k}\left(x\right)+1-p_{k}\right)-x\\
\geq & \left(1-p_{j}\right)\prod_{k\neq i,j}\left(p_{k}F_{k}\left(\underline{s}_{j}\right)d-p_{k}\right)-x\\
= & \lambda
\end{array}
\end{equation}
Where the inequality holds due to the fact that $F_{i}\left(x\right)>0$. 
\end{itemize}
That is, in both cases there exists $x$ such that $\pi_{j}\left(x\right)>\lambda$,
a contradiction to Lemma~\ref{lem:lambda}.
\end{proof}
\begin{lem}
\label{lem:eq}If $p_{i}=p_{j}$ then $\underline{s}_{i}=\underline{s}_{j}$\end{lem}
\begin{proof}
Immediate from Lemma \ref{lem:sisj}.
\end{proof}

\begin{lem}
For every $j$ such that $p_{j}=p_{n}$, $j\in\arg\min_{i}\underline{s}_{i}$.
\label{lem:maxp}\end{lem}
\begin{proof}
Immediate from Lemma \ref{lem:eq} and Lemma \ref{lem:sisj}.\end{proof}
\begin{lem}
\textup{$\min_{i}\underline{s}_{i}=0$.\label{lem:0}}\end{lem}
\begin{proof}
Otherwise, 
let $j\in\arg\min_{i}\underline{s}_{i}$. From Lemma~\ref{lem:supportEq2} we have that $\pi_j \left(\underline{s}_{j}\right) = \pi^*_j$. As for every $x \in \left(0,\underline{s}_{j}\right)$: 
$\Pr\left(j\,\text{wins}\mid \underline{s}_{j} \right) = \Pr\left(j\,\text{wins}\mid x\right)$,
it pays for $j$ to transfer mass from an 
$\varepsilon$-neighborhood above $\underline{s}_{j}$
to an $\delta$-neighborhood below $\underline{s}_{j}$.
\end{proof}

\begin{lem}
\textup{$\left|\arg\min_{i}\underline{s}_{i}\right|\geq2$. \label{lem:2}}\end{lem}
\begin{proof}
Assume to the contrary that $\left|\arg\min_{i}\underline{s}_{i}\right|=1$,
 let $\{j\} = \arg\min_{i}\underline{s}_{i}$,
and let $k \in \arg\min_{i \neq j}\underline{s}_{i}$.
From Lemma \ref{lem:0}
we have that $\underline{s}_{j}=0$.
There exists $x>0$ in bidder $j$'s
support such that $x<\underline{s}_{k}$. $x > 0$
and in bidder $j$'s support thus $\pi_{j}\left(x\right)=\lambda$.
Since $\Pr\left(j\,\text{wins}\mid x\right)=\Pr\left(j\,\text{wins}\mid y\right)$
for $y\in\left(0,x\right)$, 
it pays for $j$ to transfer mass from an 
$\varepsilon$-neighborhood of $x$ 
to some $\delta$-neighborhood of $y$.
Hence we must have that $\left|\arg\min_{i}\underline{s}_{i}\right|\geq2$
and the claim follows.
\end{proof}

\begin{lem}
If bidder $i$ has an atomic point at $0$, then $p_i=p_{i+1}=\ldots=p_n$.\label{lem:pnpi}\end{lem}
\begin{proof}
From Lemma \ref{lem:pn-1} there is no other bidder with an atomic point
at $0$. Thus, the expected profit of bidder $i$ from bidding $0$
is $\lambda=\prod_{j\neq i}\left(1-p_{j}\right)$, since $\lambda\geq\prod_{j=1}^{n-1}\left(1-p_{j}\right)$
and $p_{i}\leq p_{n}$, it must hold that $p_{i}=p_{n}$.\end{proof}
\begin{lem}
If $p_{i}<p_{j}$ and $p_{j}<p_{n}$, then $\underline{s}_{i}>\underline{s}_{j}$.
\label{lem:sOreder}\end{lem}
\begin{proof}
From Lemma \ref{lem:sisj} it holds that $\underline{s}_{i}\geq\underline{s}_{j}$.
Suppose that $\underline{s}_{i}=\underline{s}_{j}$, according to the assumption $p_{i}<p_{j}<p_{n}$
 from Lemma \ref{lem:pnpi} it holds that $\alpha_{i}\left(\underline{s}_{i}\right)=\alpha_{j}\left(\underline{s}_{j}\right)=0$.
There exists $\varepsilon>0$ such that for every $x\in\left(\underline{s}_{i},\underline{s}_{i}+\varepsilon\right)$: $x$ is in the support of bidders $i$ and $j$. The expected profit
of both bidders from bidding $x\in\left(\underline{s}_{i},\underline{s}_{i}+\varepsilon\right)$
is $\lambda$, hence 
\begin{equation}
\lim_{x\rightarrow\underline{s}_{i}+}\pi_{i}\left(x\right)=\lim_{x\rightarrow\underline{s}_{i}+}\left(p_{j}F_{j}\left(x\right)+1-p_{j}\right)\prod_{k\neq i,j}\left(p_{k}F_{k}\left(x\right)+1-p_{k}\right)-x=\lambda
\end{equation}
and 
\begin{equation}
\lim_{x\rightarrow\underline{s}_{j}+}\pi_{j}\left(x\right)=\lim_{x\rightarrow\underline{s}_{j}+}\left(p_{i}F_{i}\left(x\right)+1-p_{i}\right)\prod_{k\neq i,j}\left(p_{k}F_{k}\left(x\right)+1-p_{k}\right)-x=\lambda
\end{equation}
$\alpha_{i}\left(\underline{s}_{i}\right)=\alpha_{j}\left(\underline{s}_{j}\right)=0$,
hence 
\begin{equation}
\begin{array}{rl}
\lim_{x\rightarrow\underline{s}_{i}^{+}}\pi_{i}\left(x\right)= & \left(p_{j}F_{j}\left(\underline{s}_{i}\right)+1-p_{j}\right)\lim_{x\rightarrow\underline{s}_{i}^{+}}\prod_{k\neq i,j}\left(p_{k}F_{k}\left(x\right)+1-p_{k}\right) - \lim_{x\rightarrow\underline{s}_{i}^{+}} x\\
= & \left(1-p_{j}\right)\lim_{x\rightarrow\underline{s}_{i}^{+}}\prod_{k\neq i,j}\left(p_{k}F_{k}\left(x\right)+1-p_{k}\right) - \underline{s}_{i}
\end{array}
\end{equation}
and 
\begin{equation}
\lim_{x\rightarrow\underline{s}_{j}^{+}}\pi_{j}\left(x\right)=\left(1-p_{i}\right)\lim_{x\rightarrow\underline{s}_{j}^{+}}\prod_{k\neq i,j}\left(p_{k}F_{k}\left(x\right)+1-p_{k}\right)- \underline{s}_{j}
\end{equation}
 as $p_{i}<p_{j}$ it holds that $\lambda\neq\lambda$, thus $\underline{s}_{i}>\underline{s}_{j}$
and the claim follows.
\end{proof}

\begin{lem}
$\lambda=\prod_{j=1}^{n-1}\left(1-p_{j}\right)$\label{lem:lam1}.\end{lem}
\begin{proof}
From Lemma \ref{lem:sisj} and Lemma \ref{lem:0} it holds that $\underline{s}_n=0$, and from Lemma \ref{lem:pnpi} we can assume without loss of generality that for every $i<n$, $\alpha_i(0)=0$. 
Hence, we have that
\begin{equation}
\lambda=\pi_{n}^{*}=\lim_{x\rightarrow0^{+}}\pi_{n}\left(x\right)=\lim_{x\rightarrow0^{+}}\prod_{j=1}^{n-1}\left(p_{j}F_{j}\left(x\right)+1-p_{j}\right)=\prod_{j=1}^{n-1}\left(1-p_{j}\right).
\end{equation}
\end{proof}

\begin{lem}
There exists $i$ and $j$ such that $\overline{s}_{i}=\overline{s}_{j}=1-\lambda$.\end{lem}
\begin{proof}
Suppose not. Clearly no bidder has an incentive to bid more then $1-\lambda$.
If there is only one bidder with the highest $\overline{s}$ then they have an
incentive to reduce $\overline{s}$. If there are two bidders with the
highest $\overline{s}$, say $i$ and $j$, and $\overline{s}_{i}=\overline{s}_{j}<1-\lambda$,
then it pays for one of them to transfer mass from an 
	$\varepsilon$-neighborhood below $x$ 
	to some $\delta$-neighborhood above $x$.
 \end{proof}
\begin{lem}
For $x>0$, $i$ and $j$, if $x$ is in the support of bidders $i$
and $j$, then $p_{j}F_{j}\left(x\right)+1-p_{j}=p_{i}F_{i}\left(x\right)+1-p_{i}$
. \label{lem:cdfeq}\end{lem}
\begin{proof}
Let $x>0$ be in the support of bidders $i$ and $j$. Since $x>0$ From~Lemma \ref{lem:atom} it holds that no bidder has an atom point at
$x$. $x$ is in bidder $i$'s support hence,
\begin{equation}
\lambda=\left(p_{j}F_{j}\left(x\right)+1-p_{j}\right)\prod_{k\neq i,j}\left(p_{k}F_{k}\left(x\right)+1-p_{k}\right)-x.
\end{equation}
 $x$ is also in bidder $j$'s support and therefore,
\begin{equation}
\lambda=\left(p_{i}F_{i}\left(x\right)+1-p_{i}\right)\prod_{k\neq i,j}\left(p_{k}F_{k}\left(x\right)+1-p_{k}\right)-x
\end{equation}
we thus have that 
\begin{equation}
\left(p_{i}F_{i}\left(x\right)+1-p_{i}\right)\prod_{k\neq i,j}\left(p_{k}F_{k}\left(x\right)+1-p_{k}\right)=\left(p_{j}F_{j}\left(x\right)+1-p_{j}\right)\prod_{k\neq i,j}\left(p_{k}F_{k}\left(x\right)+1-p_{k}\right).
\end{equation}
Note that every term in the product is positive as $p_{n-1} <1$ and $\lambda > 0$, therefore
\begin{equation}
 p_{j}F_{j}\left(x\right)+1-p_{j}=p_{i}F_{i}\left(x\right)+1-p_{i}
\end{equation}
and the claim holds.\end{proof}
\begin{lem}
If $F_{i}$ is strictly increasing on some open interval $\left(a,b\right)$,
where $0\le a<b<1-\lambda$, then $F_{i}$ is strictly increasing
on the interval $\left(a,1-\lambda\right]$. \label{lem:stricinc}\end{lem}
\begin{proof}
Suppose not, then without loss of generality, assume that $F_{i}$ is constant over the interval
$\left(b,c\right)$ for some $c\in\left(b,1-\lambda\right]$. There
must be some $\varepsilon>0$ such that there is at least one bidder, say
$j$, with $F_{j}$ strictly increasing on $\left(b,b+\varepsilon\right)$.
(Otherwise, it pays to some other player to place an atom at $b$.) From Lemma~\ref{lem:supportEq2} $\pi_{i}\left(b\right)=\lambda$ and from Lemma \ref{lem:cdfeq},
$p_{i}F_{i}\left(b\right)+1-p_{i}=p_{j}F_{j}\left(b\right)+1-p_{j}$.
Let $x\in\left(b,b+\varepsilon\right)$, $x$ is not in bidder $i$'s
support hence $\pi_{i}\left(x\right)\leq\pi_{i}\left(b\right)=\lambda$,
$x$ is in bidder $j$'s support hence $\pi_{j}\left(x\right)=\lambda$.
$F_{j}$ is strictly increasing on $\left(b,b+\varepsilon\right)$
hence $F_{j}\left(x\right)>F_{j}\left(b\right)$, $F_{i}$ is constant
on $\left(b,b+\varepsilon\right)$ hence $F_{i}\left(b\right)=F_{i}\left(x\right)$.
Add it all together:
\begin{equation}
\begin{array}{rclc}
\pi_{i}\left(x\right) & \leq & \pi_{j}\left(x\right) & \iff\\
\prod_{k\neq i}\left(p_{k}F_{k}\left(x\right)+1-p_{k}\right)-x & \leq & \prod_{k\neq j}\left(p_{k}F_{k}\left(x\right)+1-p_{k}\right)-x & \iff\\
p_{j}F_{j}\left(x\right)+1-p_{j} & \leq & p_{i}F_{i}\left(x\right)+1-p_{i} & \iff\\
p_{j}F_{j}\left(x\right)+1-p_{j} & \leq & p_{i}F_{i}\left(b\right)+1-p_{i} & \iff\\
p_{j}F_{j}\left(x\right)+1-p_{j} & \leq & p_{j}F_{j}\left(b\right)+1-p_{j} & \iff\\
F_{j}\left(x\right) & \leq & F_{j}\left(b\right).
\end{array}
\end{equation}
A contradiction to the fact that $F_{j}\left(x\right)>F_{j}\left(b\right)$.
Hence, the claim holds. \end{proof}
\begin{lem}
For every $1 \leq i \leq n$, $\overline{s}_{i}=1-\lambda.$ \label{lem:upperBound}\end{lem}
\begin{proof}
Immediate from Lemma \ref{lem:stricinc}, we have that $\overline{s}_{i}\in\left\{ 0,1-\lambda\right\} $.
If $\overline{s}_{i}=0$, it means that $p_{i}=p_{n}$, bidder
$i$ has an atomic point at $0$ of $1$ and $\lambda = \prod_{k \neq i} \left(i-p_k\right)$.
Now, 
$\lim_{x\rightarrow0^{+}}\pi_{j}(x)=\prod_{k\neq i,j}\left(1-p_{k}\right)>\lambda$
 and therefore 
it pays for $j$ to transfer mass to a 
$\delta$-neighborhood above $0$. 
 Hence,
 for every $i$: $\overline{s}_{i}=1-\lambda$. \end{proof}
\begin{lem}
If $x>0$ is in bidder $i$'s support, then there must be $j\neq i$ such that
$x$ is in bidder $j$'s support. \label{lem:jointsup}\end{lem}
\begin{proof}
Let $x>0$ be in bidder $i$'s support. If
for every $y\in\left(0,x\right)$ and for every $j\neq i$: $y$
is not in bidder $j$'s support, then
it pays for $i$ to transfer mass from an 
$\varepsilon$-neighborhood above $x$
to a $\delta$-neighborhood below $x$.
 Hence, there must exists
$j\neq i$ and $y\in\left(0,x\right)$ such that $y$ is in bidder
$j$'s support, from Lemma \ref{lem:atom} $\alpha_{j}\left(y\right)=0$,
hence there is a neighborhood below $y$ such that $F_j$ is strictly increasing in that neighborhood.
From Lemma \ref{lem:stricinc}, $\left(y,1-\lambda\right]$ is
in bidder $j$'s support. $x$ is in bidder $i$'s support so $x\leq1-\lambda$,
hence $x$ is in bidder $j$'s support and the claim follows.\end{proof}
\begin{lem}
There are at least two bidders with strictly increasing CDFs on
$\left(0,1-\lambda\right]$.

\label{lem:full}\end{lem}
\begin{proof}
From Lemmata \ref{lem:2} and \ref{lem:pn-1}, there is $j\in\arg\min_{i}\underline{s}_{i}$
with $\alpha_{j}\left(0\right)=0$, from Lemma \ref{lem:stricinc}
$F_{j}$ is strictly increasing on $\left(0,1-\lambda\right]$, from
Lemma \ref{lem:jointsup} there is another bidder with is strictly
increasing CDF on $\left(0,1-\lambda\right]$.\end{proof}
\begin{lem}
There exists a continuous function $z:\left(0,1-\lambda\right]\rightarrow\left[0,1\right]$
such that if $x$ is in bidder $i$'s support then $F_{i}\left(x\right)=\frac{z\left(x\right)+p_{i}-1}{p_{i}}$.
\label{lem:conZ}\end{lem}
\begin{proof}
From Lemma \ref{lem:full} there is $j$, such that $F_{j}$ is strictly
increasing on $\left(0,1-\lambda\right]$. Let us define for every $x\in\left(0,1-\lambda\right]$:
$z\left(x\right):=p_{j}F_{j}\left(x\right)+1-p_{j}$. For every $i$
and $x$ in bidder $i$'s support, from Lemma \ref{lem:cdfeq} it holds that
$F_{i}\left(x\right)=\frac{z\left(x\right)+p_{i}-1}{p_{i}}$ and the
claim holds. \end{proof}
\begin{lem}
\label{lem:2pn}If $p_{n}=p_{n-1}<1$ and $F_{i}$ is strictly increasing
on the interval $\left(0,1-\lambda\right]$ then $p_{i}=p_{n}$. \end{lem}
\begin{proof}
	Since there are at least two bidders
	with $p=p_{n}$, and 	
	from Lemma \ref{lem:cdfeq} it holds that for every $j$ if $p_{j}=p_{n}$
	then $F_{j}=F_{n}$; we may conclude that there is 
	no bidder with an atomic point at $0$.
	That is,
	 $ \lim_{x\rightarrow0^{+}}F_{n}\left(x\right)=F_{n}\left(0\right)=0$.
	 $F_{i}$ is strictly increasing on $\left(0,1-\lambda\right]$, and
	 $F_{i}\left(0\right)=0$ hence $\lim_{x\rightarrow0^{+}}F_{i}\left(x\right)=F_{i}\left(0\right)=0$.
	 From Lemma \ref{lem:cdfeq} it holds that:	
\begin{equation}
\begin{array}{rclc}
p_{j}\lim_{x\rightarrow0^{+}}F_{j}\left(x\right)+1-p_{j} & = & p_{i}\lim_{x\rightarrow0^{+}}F_{i}\left(x\right)+1-p_{i} & \iff\\
p_{j}F_{j}\left(0\right)+1-p_{j} & = & p_{i}F_{i}\left(0\right)+1-p_{i} & \iff\\
1-p_{j} & = & 1-p_{i} & \iff\\
p_{i} & = & p_{j}
\end{array}
\end{equation}
and the claim follows.

\end{proof}

\begin{lem}
If $p_{n}\neq p_{n-1}<1$ bidder $n$ has an atomic point in the distribution
at $0$ of $1-\frac{p_{n-1}}{p_{n}}$.\end{lem}
\begin{proof}
From Lemma~\ref{lem:sOreder}, Lemma~\ref{lem:2} and Lemma~\ref{lem:0}, $\underline{s}_{n}=\underline{s}_{n-1}=0$.
From Lemma~\ref{lem:upperBound} and Lemma~\ref{lem:stricinc} $F_{n}$
and $F_{n-1}$ are strictly increasing on $\left(0,1-\lambda\right]$.
From Lemma \ref{lem:cdfeq} for every $x\in\left(0,1-\lambda\right]$
it holds that 
\begin{equation}
p_{n}F_{n}\left(x\right)+1-p_{n}=p_{n-1}F_{n-1}\left(x\right)+1-p_{n-1},
\end{equation}
and from Lemma~\ref{lem:pnpi} $F_{n-1}(0)=0$.
That is, 
\begin{equation}
\begin{array}{rclc}
\lim_{x\rightarrow0^{+}}p_{n}F_{n}\left(x\right)+1-p_{n} & = & \lim_{x\rightarrow0^{+}}p_{n-1}F_{n-1}\left(x\right)+1-p_{n-1} & \iff\\
p_{n}F_{n}\left(0\right)+1-p_{n} & = & p_{n-1}F_{n-1}\left(0\right)+1-p_{n-1} & \iff\\
p_{n}\alpha_{n}\left(0\right)+1-p_{n} & = & 1-p_{n-1}. 
\end{array}
\end{equation}
Hence
$\alpha_{n}\left(0\right)=1-\frac{p_{n-1}}{p_{n}}$.
\end{proof}
We are now, finally, ready to prove Theorem \ref{thm:positiveProfit}.
\begin{customthm}{\ref{thm:positiveProfit}}
	In common values all-pay auction when the item value is $1$, if $p_{n-1}<1$ then there is a
	unique Nash equilibrium, in which the expected profit of every participating
	bidder is $\prod_{j=1}^{n-1}\left(1-p_{j}\right)$. Furthermore, there
	exists a continuous function $z:\left[0,1 - \prod_{j=1}^{n-1}\left(1-p_{j}\right) \right]\rightarrow\left[0,1\right]$,
	such that when a bidder $i$, has a positive density over an interval,
	they bid according to $F_{i}\left(x\right)=\frac{z\left(x\right)+p_{i}-1}{p_{i}}$
	over that interval, and if $p_{i}=p_{j}$ then $F_{i}=F_{j}$.	
\end{customthm}

\begin{proof}
Let $0<p_{1}\leq\dots\leq p_{n}$, such that $p_{n-1}<1$. From Lemma~\ref{lem:lambda} and Lemma~\ref{lem:lam1} we have that in any Nash equilibrium
the expected profit of every participating bidder is $\lambda=\prod_{j=1}^{n-1}\left(1-p_{j}\right)$,
from Lemma \ref{lem:upperBound} we have that for every $i$: $\overline{s}_{i}=1-\lambda$,
and from Lemma \ref{lem:sisj} we have that $1-\lambda\geq\underline{s}_{1}\geq\dots\geq\underline{s}_{n}=0$. 

There are two possible cases, either $p_{n}=p_{n-1}$ or $p_{n}>p_{n-1}$;
\begin{itemize}
\item If $p_{n}=p_{n-1}$, let $m_{n}$ be the minimal index such that $p_{m_{n}}=p_{n}$.
Since $p_{n}=p_{n-1}$ there is no bidder with an atomic point at $0$.
From Lemma \ref{lem:2pn} it holds that $\underline{s}_{m_{n}-1}>\underline{s}_{m_{n}}=0$,
adding it with Lemmas~\ref{lem:sOreder} and Lemma~\ref{lem:sisj} we have
that $p_{i}<p_{j}\iff\underline{s}_{i}>\underline{s}_{j}$.

For every $i$: $F_{i}$ is strictly increasing in $\left(\underline{s}_{1},1-\lambda\right]$,
since the expected profit of each participating bidder is $\lambda$,
we have that for every $x\in \left(\underline{s}_{1},1-\lambda\right]$
\begin{equation}
\begin{array}{rccl}
\lambda= & \pi_{i}\left(x\right) & = & \prod_{j\neq1}\left(1-p_{j}+p_{j}F_{j}\left(x\right)\right)-x\\
& & = & \prod_{j\neq1}\left(1-p_{j}+p_{j}\frac{z\left(x\right)+p_{j}-1}{p_{j}}\right)-x\\
& & = & z\left(x\right)^{n-1}-x.
\end{array}
\end{equation}
Therefore for every $x \in \left(\underline{s}_{1},1-\lambda\right]$ and for every $i$:
 $z\left(x\right)=\left(x+\lambda\right)^{\frac{1}{n-1}}$ and
 $F_{i}\left(x\right)=\frac{\left(x+\lambda\right)^{\frac{1}{n-1}}+p_{i}-1}{p_{i}}$.
Since there is no bidder with an atomic point at $0$, it holds that $F_{1}\left(\underline{s}_{1}\right)=0$,
hence it follows that $\underline{s}_{1}=\left(1-p_{1}\right)^{n-1}-\lambda$.

Let $m_{2}$ be the minimal index such that $p_{m_{2}}>p_{1}$. From
Lemma \ref{lem:eq} we have that for every $j < m_{2}$ it holds that $\underline{s}_{j}=\underline{s}_{1}$
and $\underline{s}_{m_{2}}<\underline{s}_{1}$. 

For every
$i < m_{2}$ and for every $x\leq\underline{s}_{1}$, $F_{i}\left(x\right)=0$;
for every $i \geq m_{2}$ and for every $x\in\left(\underline{s}_{m_2},\underline{s}_{1}\right]$,
 $F_{i}\left(x\right)>0$ and $\pi_i\left(x\right) =\lambda$. 
 \begin{equation}
 \begin{array}{rccl}
 \lambda= & \pi_{i}\left(x\right) & = & \prod_{j\neq1}\left(1-p_{j}+p_{j}F_{j}\left(x\right)\right)-x\\
 & & = & \prod_{j<k_{2}}\left(1-p_{j}+p_{j}F_{j}\left(x\right)\right)\prod_{
 	\tiny{
 		\begin{array}{c}
 		j\geq m_{2}\\
 		j\neq i
 		\end{array}}}\left(1-p_{j}+p_{j}F_{j}\left(x\right)\right)-x\\
 & & = & \prod_{j<m_{2}}\left(1-p_{j}\right)\prod_{
 	\tiny{
 	\begin{array}{c}
 	j\geq m_{2}\\
 	j\neq i
 	\end{array}}}\left(1-p_{j}+p_{j}\frac{z\left(x\right)+p_{j}-1}{p_{j}}\right)-x\\
 & & = & \left(1-p_{j}\right)^{m_{2}-1}z\left(x\right)^{n-m_{2}}-x,
 \end{array}
 \end{equation}
 Therefore for every $x\in\left(\underline{s}_{m_2},\underline{s}_{1}\right]$ and for every $i \geq m_2$:
$z\left(x\right)=\left(\frac{x+\lambda}{\left(1-p_{j}\right)^{m_{2}-1}}\right)^{\frac{1}{n-m_2}}$
 and
$F_{i}\left(x\right)=\frac{\left(\frac{x+\lambda}{\left(1-p_{j}\right)^{m_{2}-1}}\right)^{\frac{1}{n-m_2}}+1-p_{i}}{p_{i}}$.
$F_{m_2}\left(\underline{s}_{m_2}\right)=0$, hence it follows
that $\underline{s}_{m_2}=\left(1-p_{m_2}\right)^{n-m_2}\left(1-p_{1}\right)^{m_2-1}-\lambda$.

In the general case, let $k$ be an index such that $\underline{s}_{k-1}>\underline{s}_{k}$,
that is $p_{k-1}<p_{k}$. Since for every $i<k$ and for every $x\in\left(\underline{s}_{k},\underline{s}_{k-1}\right]$:
$F_{j}\left(x\right)=0$; for every $i\geq k$ and for every $x\in\left(\underline{s}_{k},\underline{s}_{k-1}\right]$:
 $F_{i}\left(x\right)>0$ and $\pi_i(x) = \lambda$. 
 \begin{equation}
 \begin{array}{rccl}
 	\lambda= & \pi_{i}\left(x\right) & = & \prod_{j\neq1}\left(1-p_{j}+p_{j}F_{j}\left(x\right)\right)-x\\
 	& & = & \prod_{j<k}\left(1-p_{j}+p_{j}F_{j}\left(x\right)\right)\prod_{
 		\tiny{
 		\begin{array}{c}
 			j\geq k\\
 			j\neq i
 		\end{array}}}\left(1-p_{j}+p_{j}F_{j}\left(x\right)\right)-x\\
 		& & = & \prod_{j<k}\left(1-p_{j}\right)\prod_{
 			\tiny{
 			\begin{array}{c}
 				j\geq k\\
 				j\neq i
 			\end{array}}}\left(1-p_{j}+p_{j}\frac{z\left(x\right)+p_{j}-1}{p_{j}}\right)-x\\
 			& & = & \prod_{j=1}^{k-1}\left(1-p_{j}\right)z\left(x\right)^{n-k}-x.
 		\end{array}
 \end{equation}
 Therefore for every $x\in\left(\underline{s}_{k},\underline{s}_{k-1}\right]$ and for every $i \geq k$:
 $z\left(x\right)=\left(\frac{x+\lambda}{\prod_{j=1}^{k-1}\left(1-p_{j}\right)}\right)^{\frac{1}{n-k}}$ 
 and
 $F_{i}\left(x\right)=\frac{\left(\frac{x+\lambda}{\prod_{j=1}^{k-1}\left(1-p_{j}\right)}\right)^{\frac{1}{n-k}}+1-p_{i}}{p_{i}}$.
 $F_{k}\left(\underline{s}_{k}\right)=0$, hence it follows that
 $\underline{s}_{k}=\left(1-p_{k}\right)^{n-k}\prod_{j=1}^{k-1}\left(1-p_{j}\right) -\lambda$.

We can continue in the same way, until the interval $\left(0,\underline{s}_{m_{n}}\right]$,
in which for every $i < m_{n}$ and for every $x \in\left(0,\underline{s}_{m_{n}}\right]$:
 $F_{j}\left(x\right)=0$; and for
every $i \geq m_{n}$ and for every $x \in \left(0,\underline{s}_{m_{n}}\right]$:
$F_{j}\left(x\right)>0$. Therefore 
for every $i \geq m_{n}$ and for every $x\in\left(0,\underline{s}_{m_{n}}\right]$:
 $F_{i}\left(x\right)=\frac{\left(\frac{\left(x+\lambda\right)}{\prod_{j=1}^{m_n-1}\left(1-p_{j}\right)}\right)^{\frac{1}{n-m_{n}}}+p_{i}-1}{p_{i}}$.
 
Thus, for every $i$: $F_{i}$ is uniquely determined, and therefore
the equilibrium is unique.

\item If $p_{n}>p_{n}$, similarly to the previous case, for every $i$:
 $F_{i}$ is uniquely determined. The only exceptions are that $\underline{s}_{n}=\underline{s}_{n-1}$,
and bidder $n$ has an atomic point at $0$, of $1-\frac{p_{n-1}}{p_{n}}$.
\end{itemize}
Hence, for $0<p_{1}\leq p_{2}\leq\dots\leq p_{n-1}\leq p_{n}$, where
$p_{n-1}<1$, and either $p_{n}=p_{n-1}$ or $p_{n}>p_{n-1}$, in
the unique equilibrium, $\underline{s}_{1}=\left(1-p_{1}\right)^{n-1}-\lambda$,
$\underline{s}_{k}=\left(1-p_{k}\right)^{n-k}\prod_{j=1}^{k-1}\left(1-p_{j}\right)-\lambda$
for $k\in\left\{ 2,\dots,n-1\right\}$, and $\underline{s}_{n}=\underline{s}_{n-1}=0$,
where $\lambda=\prod_{j=1}^{n-1}\left(1-p_{j}\right)$. For every $1 \leq i \leq n-1 $ the CDFs
are: 
\begin{equation}
F_{i}\left(x\right)=\begin{cases}
1 & \text{\ensuremath{x\geq1-\lambda}}\\
\frac{z\left(x\right)+p_{i}-1}{p_{i}} & x\in\left[\underline{s}_{1},1-\lambda\right)\\
\vdots & \vdots\\
\frac{z\left(x\right)+p_{i}-1}{p_{i}} & x\in\left[\underline{s}_{k},\underline{s}_{k-1}\right)\\
\vdots & \vdots\\
\frac{z\left(x\right)+p_{i}-1}{p_{i}} & x\in\left[\underline{s}_{i},\underline{s}_{i-1}\right)\\
0 & x<\underline{s}_{i}
\end{cases}
\end{equation}
 and $F_{n}=\frac{p_{n-1}}{p_{n}}F_{n-1}$,
where 
\begin{equation}
z\left(x\right)=\begin{cases}
\left(\lambda+x\right)^{\frac{1}{n-1}} & x\in\left[\underline{s}_{1},1-\lambda\right)\\
\vdots & \vdots\\
\left(\frac{\lambda+x}{\prod_{j=1}^{k-1}\left(1-p_{j}\right)}\right)^{\frac{1}{n-k}} & x\in\left[\underline{s}_{k},\underline{s}_{k-1}\right)\\
\vdots & \vdots\\
\frac{\lambda+x}{\prod_{j=1}^{n-2}\left(1-p_{j}\right)} & x\in\left[0,\underline{s}_{n-2}\right).
\end{cases}
\end{equation}
\end{proof}

\subsection*{Proof of Theorem~\ref{thm:zeroProfit} \label{subsec:thmzero}}
When there are at least two bidders with $p=1$, the auction approaches
the case without failures. If we do not allow agent failures, \cite{ref1}
characterized
the equilibria. 
\begin{thm}[\cite{ref1}]
	The first price sealed bid all pay common values auction with complete
	information possesses two types of equilibria. Either all players
	use the same continuous mixed strategy with support $\left[0,1\right]$,
	or at least two bidders randomize continuously over $\left[0,1\right]$
	with each other player $i$ randomizing continuously over $\left(b_{i},1\right]$,
	$b_{i}>0$, and having an atomic at $0$ equals to $G_{i}\left(b_{i}\right)$.
	When two or more bidders have a positive density over a common interval
	they play the same continuous mixed strategy over that interval. Furthermore,
	if bidders $m+1,\dots,n$ randomized continuously over $\left[0,1\right]$,
	with bidder $i=1,\dots,m$ randomized continuously over $\left(b_{i},1\right]$,
	with $1\geq b_{1}\geq\dots\geq b_{m}$. The equilibrium strategies
	are:
	
	\begin{equation}
	G_{i}\left(x\right)=\begin{cases}
	x^{\frac{1}{n-1}} & x\in\left[b_{1},1\right]\\
	\left(\frac{x}{G_{1}\left(b_{1}\right)}\right)^{\frac{1}{n-2}} & x\in\left[b_{2},b_{1}\right)\\
	\vdots & \vdots\\	
	
	\left(\frac{x}{\prod_{j<k}G_{j}\left(b_{j}\right)}\right)^{\frac{1}{k-1}} & x\in\left[b_{k},b_{k-1}\right)\\
	\vdots & \vdots\\
	\left(\frac{x}{\prod_{j<i}G_{j}\left(b_{j}\right)}\right)^{\frac{1}{i-1}} & x\in\left[b_{i},b_{i-1}\right)\\
	G_{i}\left(b_{i}\right) & x\in\left[0,b_{i}\right)\\
	0 & x<0	
	
	\end{cases}
	\end{equation}
	\label{thm:baye}
\end{thm}
In the Nash equilibrium strategies presented in \cite{ref1},
the $b_{i}$'s fully defines the atomic point of bidder $i$ at $0$,
and vice versa.\footnote{We can set the atomic point of bidder $1$ at $0$, $\alpha_{1}(0)$, to $b_{1}^{\frac{1}{n-1}}$, and $\alpha_{2}\left(0\right)=\left(\frac{b_{2}}{\alpha_{1\left(0\right)}}\right)^{\frac{1}{n-2}}$,...,$\alpha_{m}\left(0\right)=\left(\frac{b_{m}}{\prod_{j<m}\alpha_{j}\left(0\right)}\right)$.
	In the other direction, once we set $\alpha_{i}\left(0\right)$, we
	can set $b_{1}=\alpha_{1}\left(0\right)^{n-1}$, $b_{2}=\alpha_{1}\left(0\right)\alpha_{2}\left(0\right)^{n-2}$,...,$b_{m}\left(0\right)=\prod_{j<m}\alpha_{j}\left(0\right)\alpha_{m}\left(0\right)^{n-m}$.%
}
That is, every Nash equilibrium is defined by the bidders' atomic points
at $0$, and when the bidders' atomic points at $0$ are set, the Nash
equilibrium is unique. The expected profit of every bidder in any Nash
equilibrium is $0$.

\begin{customthm}{\ref{thm:zeroProfit}}
	In common values all-pay auction when the item value is $1$, if $p_{n-1}=1$ then in every Nash
	equilibrium the expected profit of every participating bidder is $0$.
	At least two bidders with $p=1$ randomize over $\left[0,1\right]$
	with each other player $i$ randomizing continuously over $\left(b_{i},1\right]$,
	$b_{i}>0$, and having an atomic point at $0$ of $\alpha_{i}\left(0\right)$.
	There exists a continuous function $z\left(x\right):\left[0,1\right]\rightarrow\left[0,1\right]$,
	such that when a bidder, $i$, has a positive density over an interval,
	they bid according to $F_{i}\left(x\right)=\frac{z\left(x\right)+p_{i}-1}{p_{i}}$
	over that interval. For every $i$ --- the atomic point at $0$
	is equals to $F_{i}\left(0\right)$. 
\end{customthm}

\begin{proof}
	Let $0<p_{1}\leq\dots p_{n-2}\leq p_{n-1}=p_{n}=1$, let $F_{1},\dots F_{n}$,
	be a Nash equilibrium, and let $\alpha_{i}\left(0\right)$ be bidder
	$i$'s atomic point at $0$, that is $\alpha_{i}\left(0\right)=F_{i}\left(0\right)$.
	We first note that there must be a bidder with $p=1$ and $\alpha_{i}\left(0\right)=0$,
	otherwise any other bidder with $p=1$ has an incentive to bid a small
	$\varepsilon$ and win the item with a positive probability. Hence, failing to participate in the auction and biding
	$0$, would lead to the same result --- $0$ probability of winning
	the item. 
	
	If a bidder has a participating probability of $p$, and an atomic point
	at $0$ of $\alpha\left(0\right)$, it can be considered as the bidder
	has a participating probability of $1$, and an atomic point at $0$
	of $1-p\left(1-\alpha\left(0\right)\right)$. The probability to outbid
	bidder $i$ when bidding $x$, is the probability that either bidder
	$i$ had failed or bidder $i$ hadn't failed and bid less then $x$. That is, the probability
	of outbidding bidder $i$ is $1-p_{i}+p_{i}F_{i}\left(x\right)$,
	hence if we let $\beta_{i}\left(0\right)=1-p_{i}\left(1-\alpha_{i}\left(0\right)\right)$
	and $G_{i}\left(x\right)=1-p_{i}+p_{i}F_{i}\left(x\right)$, then
	the $G_{i}$s are equilibrium distributions of the auction without
	agent failures, in which bidder $i$ has an atomic point at $0$ of $\beta_{i}\left(0\right)$,
	and the expected profit of each bidder is $0$.
	
	From Theorem \ref{thm:baye}, at least two bidders randomized continuously
	over $\left[0,1\right]$ with no atomic point at $0$, that is, for
	at least two bidders $\beta_{i}\left(0\right)=0$, hence it must hold
	that for at least two bidders, say $n$ and $n-1$, with $p=1$, $\alpha\left(0\right)=0$.
	By letting $z\left(x\right)=G_{n}\left(x\right)$ the theorem is proved.
\end{proof}

\subsection*{Proof of Theorem~\ref{thm:expectedBid} \label{subsec:app2}}
\begin{customthm}{\ref{thm:expectedBid}}
For every $1 \leq i \leq n-1$: 
\begin{equation}
\begin{split}\mathbb{E}\left[bid_{i}\right]= & \frac{1}{p_{i}}\left(\frac{1}{n}+\sum_{k=1}^{i}\frac{\left(1-p_{k}\right)^{n-k}\prod_{j=1}^{k}\left(1-p_{j}\right)}{\left(n-k\right)\left(n-k+1\right)}-\right.\\
 & \qquad \left.-\frac{\left(1-p_{i}\right)^{n-i}\prod_{j=1}^{i}\left(1-p_{j}\right)}{n-i}-p_{i}\lambda\right)
\end{split}
\end{equation}
and 
\begin{equation}
\mathbb{E}\left[bid_{n}\right]=\frac{p_{n-1}}{p_{n}}\mathbb{E}\left[bid_{n-1}\right].
\end{equation}
\end{customthm}

\begin{proof}
For every $ 1 \leq i \leq n - 1$ we have that
\begin{equation}
\mathbb{E}\left[bid_{i}\right]=\sum_{k=1}^{i}\int\limits _{\underline{s}_{k}}^{\underline{s}_{k-1}}x f_{i}\left(x\right)\,\mathrm{d}x.
\end{equation}
Fix $k\in\left\{ 1,\dots,i\right\} $
\begin{equation}
\begin{array}{rl}
\int\limits _{\underline{s}_{k}}^{\underline{s}_{k-1}}x f_{i}\left(x\right)\,\mathrm{d}x= & \frac{1}{p_{i}}\frac{1}{n-k}\prod\limits_{j=0}^{k-1}\left(1-p_{j}\right)^{\frac{-1}{n-k}}\int_{\underline{s}_{k}}^{\underline{s}_{k-1}}x \left(\lambda+x\right)^{\frac{1-n+k}{n-k}}\,\mathrm{d}x\\
= & \left.\frac{1}{p_{i}}\prod\limits_{j=0}^{k-1}\left(1-p_{j}\right)^{\frac{-1}{n-k}}\left(\lambda+x\right)^{\frac{1}{n-k}}\left(\frac{\lambda+x}{n-k+1}-\lambda\right)\right|_{\underline{s}_{k}}^{\underline{s}_{k-1}}
\end{array}
\end{equation}
\begin{equation}
\begin{array}{rl}
\qquad \qquad \qquad= & \frac{1}{p_{i}}\left(\prod\limits_{j=0}^{k-1}\left(1-p_{j}\right)^{\frac{-1}{n-k}}\left(\lambda+\underline{s}_{k-1}\right)^{\frac{1}{n-k}}\left(\frac{\lambda+\underline{s}_{k}}{n-k+1}-\lambda\right)\right.\\
 & \left.-\prod\limits_{j=0}^{k-1}\left(1-p_{j}\right)^{\frac{-1}{n-k}}\left(\lambda+\underline{s}_{k}\right)^{\frac{1}{n-k}}\left(\frac{\lambda+\underline{s}_{k-1}}{n-k+1}-\lambda\right)\right)\\
= & {\scriptstyle \frac{1}{p_{i}}\left(\prod\limits_{j=0}^{k-1}\left(1-p_{j}\right)^{\frac{-1}{n-k}}\left(\left(1-p_{k-1}\right)^{n-k}\prod\limits_{j=0}^{k-1}\left(1-p_{j}\right)\right)^{\frac{1}{n-k}}\left(\frac{\left(1-p_{k-1}\right)^{n-k}\prod\limits_{j=0}^{k-1}\left(1-p_{j}\right)}{n-k+1}-\lambda\right)\right.}\\
 & {\scriptstyle \left.-\prod\limits_{j=0}^{k-1}\left(1-p_{j}\right)^{\frac{-1}{n-k}}\left(\left(1-p_{k}\right)^{n-k}\prod\limits_{j=0}^{k-1}\left(1-p_{j}\right)\right)^{\frac{1}{n-k}}\left(\frac{\left(1-p_{k}\right)^{n-k}\prod\limits_{j=0}^{k-1}\left(1-p_{j}\right)}{n-k+1}-\lambda\right)\right)}\\
= & {\scriptstyle \frac{1}{p_{i}}\left(\left(1-p_{k-1}\right)\left(\frac{\left(1-p_{k-1}\right)^{n-k}\prod\limits_{j=0}^{k-1}\left(1-p_{j}\right)}{n-k+1}-\lambda\right)-\left(1-p_{k}\right)\left(\frac{\left(1-p_{k}\right)^{n-k}\prod\limits_{j=0}^{k-1}\left(1-p_{j}\right)}{n-k+1}-\lambda\right)\right)}.
\end{array}
\end{equation}
Hence, for $ 1 \leq i \leq n-1 $ 
\begin{equation}
\begin{array}{rl}
\mathbb{E}\left[bid_{i}\right]= & \sum\limits_{k=1}^{i}\int\limits _{\underline{s}_{k}}^{\underline{s}_{k-1}}x f_{i}\left(x\right)\,\mathrm{d}x\\
= & \frac{1}{p_{i}}\left(\sum\limits_{k=0}^{i-1}\left(1-p_{k}\right)\left(\frac{\left(1-p_{k}\right)^{n-k-1}\prod_{j=0}^{k}\left(1-p_{j}\right)}{n-k}-\lambda\right)\right.\\
 & \left.-\sum\limits_{k=1}^{i}\left(1-p_{k}\right)\left(\frac{\left(1-p_{k}\right)^{n-k}\prod_{j=0}^{k-1}\left(1-p_{j}\right)}{n-k+1}-\lambda\right)\right)\\
= & \frac{1}{p_{i}}\left(\frac{1}{n}-\lambda+\sum\limits_{k=1}^{i}\frac{\left(1-p_{k}\right)^{n-k}\prod_{j=0}^{k}\left(1-p_{j}\right)}{\left(n-k\right)\left(n-k+1\right)}\right.\\
 & \left.-\left(1-p_{i}\right)\left(\frac{\left(1-p_{i}\right)^{n-i}\prod_{j=0}^{i-1}\left(1-p_{j}\right)}{n-i}-\lambda\right)\right)\\
= & \frac{1}{p_{i}}\left(\frac{1}{n}+\sum\limits_{k=1}^{i}\frac{\left(1-p_{k}\right)^{n-k}\prod_{j=1}^{k}\left(1-p_{j}\right)}{\left(n-k\right)\left(n-k+1\right)}-\frac{\left(1-p_{i}\right)^{n-i}\prod_{j=1}^{i}\left(1-p_{j}\right)}{n-i}-p_{i}\lambda\right).
\end{array}
\end{equation}
As $f_{n}\left(x\right)=\frac{p_{n-1}}{p_{n}}f_{n-1}\left(x\right)$ it holds that
$
\mathbb{E}\left[bid_{n}\right]=\frac{p_{n-1}}{p_{n}}\mathbb{E}\left[bid_{n-1}\right]
$.
\end{proof}

\subsection*{Proof of Theorem~\ref{thm:sumProfit} \label{subsec:app3}}

\begin{customthm}{\ref{thm:sumProfit}}
The sum-profit auctioneer's profits are:
$$
\sum_{i=1}^{n}p_{i}\mathbb{E}\left[bid_{i}\right]=1-\lambda\left(1+\sum_{i=1}^{n-1}p_{i}\right)
$$
\end{customthm}

\begin{proof}

The sum-profit auctioneer collects all the bids of the bidders that have not failed. Therefore:
\begin{equation}
\begin{array}{rl}
\mathbb{E}\left[AP\right] =& \sum\limits _{i=1}^{n}p_{i}\mathbb{E}\left[bid_{i}\right] \\
= & \sum\limits _{i=1}^{n-1}p_{i}\mathbb{E}\left[bid_{i}\right]+p_{n}\mathbb{E}\left[bid_{n}\right]\\
= & \sum\limits _{i=1}^{n-1}\left(\frac{1}{n}+\sum\limits _{k=1}^{i}\frac{\left(1-p_{k}\right)^{n-k}\prod_{j=1}^{k}\left(1-p_{j}\right)}{\left(n-k\right)\left(n-k+1\right)}-\frac{\left(1-p_{i}\right)^{n-i}\prod_{j=1}^{i}\left(1-p_{j}\right)}{n-i}-p_{i}\lambda\right)\\
 & +\frac{1}{n}+\sum\limits _{k=1}^{n-1}\frac{\left(1-p_{k}\right)^{n-k}\prod_{j=1}^{k}\left(1-p_{j}\right)}{\left(n-k\right)\left(n-k+1\right)}-\lambda\\
= & 1+\sum\limits _{k=1}^{n-1}\frac{\left(1-p_{k}\right)^{n-k}\prod_{j=0}^{k}\left(1-p_{j}\right)}{n-k}-\sum\limits _{i=1}^{n-1}\frac{\left(1-p_{i}\right)^{n-i}\prod_{j=0}^{i}\left(1-p_{j}\right)}{n-i}\\
 & -\lambda-\sum\limits _{i=1}^{n-1}p_{i}\lambda\\
= & 1-\lambda\left(1+\sum\limits _{i=1}^{n-1}p_{i}\right).
\end{array}
\end{equation}
\end{proof}

\subsection*{Proof of Theorem~\ref{thm:maxProfit} \label{subsec:app4}}

\begin{customthm}{\ref{thm:maxProfit}}
The max-profit auctioneer's profits are:
\begin{equation*}
\begin{split}
\mathbb{E}\left[AP\right] = & \int\limits _{\underline{s}_{n-1}}^{\underline{s}_{0}}x g\left(x\right)\,\mathrm{d}x\\
= & \frac{n}{2n-1}-\lambda+ \sum_{k=1}^{n-1}\left(\frac{(1-p_{k})^{2n-2k-1}\prod_{j=1}^{k}(1-p_{j})^{2}}{4(n-k)^{2}-1}\right)
\end{split}
\end{equation*}
\end{customthm}
\begin{proof}
Looking for the expected profit, we have:
$$
\begin{array}{rl}
\int\limits _{\underline{s}_{n-1}}^{\underline{s}_{0}}xg\left(x\right)\,\mathrm{d}x= & \sum\limits _{k=1}^{n-1}\int_{\underline{s}_{k}}^{\underline{s}_{k-1}}x\frac{n-k+1}{n-k}\left(\frac{\lambda+x}{\prod_{j=0}^{k-1}\left(1-p_{1}\right)}\right)^{\frac{1}{n-k}}\,\mathrm{d}x\\
= & {\scriptstyle \left.\sum\limits _{k=1}^{n-1}\left(\prod\limits _{j=0}^{k-1}\left(1-p_{j}\right)^{\frac{-1}{n-k}}\left(\lambda+x\right)^{\frac{n-k+1}{n-k}}\left(\left(\lambda+x\right)\frac{n-k+1}{2n-2k+1}-\lambda\right)\right)\right|_{\underline{s}_{k}}^{\underline{s}_{k-1}}}\\
= & {\scriptstyle \sum\limits _{k=1}^{n-1}\left(\prod\limits _{j=0}^{k-1}\left(1-p_{j}\right)\left(1-p_{k-1}\right)^{n-k+1}\left(\prod\limits _{j=0}^{k-1}\left(1-p_{j}\right)\left(1-p_{k-1}\right)^{n-k}\frac{n-k+1}{2n-2k+1}-\lambda\right)\right.}\\
 & {\scriptstyle \left.-\prod\limits _{j=0}^{k-1}\left(1-p_{j}\right)\left(1-p_{k}\right)^{n-k+1}\left(\prod\limits _{j=0}^{k-1}\left(1-p_{j}\right)\left(1-p_{k}\right)^{n-k}\frac{n-k+1}{2n-2k+1}-\lambda\right)\right)}\\
= & {\scriptstyle +\sum\limits _{k=0}^{n-2}\left(\prod\limits _{j=0}^{k}\left(1-p_{j}\right)\left(1-p_{k}\right)^{n-k}\left(\prod\limits _{j=0}^{k}\left(1-p_{j}\right)\left(1-p_{k}\right)^{n-k-1}\frac{n-k}{2n-2k-1}-\lambda\right)\right)}\\
 & {\scriptstyle -\sum\limits _{k=1}^{n-1}\left(\prod\limits _{j=0}^{k}\left(1-p_{j}\right)\left(1-p_{k}\right)^{n-k}\left(\prod\limits _{j=0}^{k}\left(1-p_{j}\right)\left(1-p_{k}\right)^{n-k-1}\frac{n-k+1}{2n-2k+1}-\lambda\right)\right)}\\
= & \frac{n}{2n-1}-\lambda+\sum\limits _{k=1}^{n-1}\frac{\left(1-p_{k}\right)^{2n-2k-1}\prod_{j=1}^{k}\left(1-p_{j}\right)^{2}}{4\left(n-k\right)^{2}-1}.
\end{array}
$$
\end{proof}

\section*{Example~\ref{exm:exm}}

\begin{customexm}{\ref{exm:exm}}
	Consider how four bidders interact. Our bidders have participation probability of
	$p_{1}=\frac{1}{3}$, $p_{2}=\frac{1}{2}$, $p_{3}=\frac{3}{4}$
	and $p_{4}=1$. Let us look at each bidder's CDFs:	
	$$
	\begin{array}{cl}
	F_{1}(x)= & \begin{cases}
	1 & x\geq\frac{11}{12}\\
	3\left(\frac{1}{12}+x\right)^{\frac{1}{3}}-2\quad\enskip & x\in\left[\frac{23}{108},\frac{11}{12}\right)\\
	0 & x<\frac{23}{108}
	\end{cases}\\
	\\
	F_{2}(x)= & \begin{cases}
	1 & x\geq\frac{11}{12}\\
	2\left(\frac{1}{12}+x\right)^{\frac{1}{3}}-1 & x\in\left[\frac{23}{108},\frac{11}{12}\right)\\
	2\left(\frac{3\left(\frac{1}{12}+x\right)}{2}\right)^{\frac{1}{2}}-1\; & x\in\left[\frac{1}{12},\frac{23}{108}\right)\\
	0 & x<\frac{1}{12}
	\end{cases}\\
	\\
	F_{3}(x)= & \begin{cases}
	1 & x\geq\frac{11}{12}\\
	\frac{4}{3}\left(\frac{1}{12}+x\right)^{\frac{1}{3}}-\frac{1}{3} & x\in\left[\frac{23}{108},\frac{11}{12}\right)\\
	\frac{4}{3}\left(\frac{3\left(\frac{1}{12}+x\right)}{2}\right)^{\frac{1}{2}}-\frac{1}{3} & x\in\left[\frac{1}{12},\frac{23}{108}\right)\\
	4\left(\frac{1}{12}+x\right)-\frac{1}{3} & x\in\left[0,\frac{1}{12}\right)\\
	0 & x<0
	\end{cases}\\
	\\
	F_{4}(x)= & \begin{cases}
	1 & x\geq\frac{11}{12}\\
	\left(\frac{1}{12}+x\right)^{\frac{1}{3}} & x\in\left[\frac{23}{108},\frac{11}{12}\right)\\
	\left(\frac{3\left(\frac{1}{12}+x\right)}{2}\right)^{\frac{1}{2}}\qquad\enskip & x\in\left[\frac{1}{12},\frac{23}{108}\right)\\
	3\left(\frac{1}{12}+x\right) & x\in\left(0,\frac{1}{12}\right)\\
	\frac{1}{4} & x=0\\
	0 & x<0
	\end{cases}
	\end{array}
	$$	
	\begin{figure}
		\begin{centering}
			\subfloat[\label{fig:examCDF}]{\includegraphics[width=0.5\textwidth]{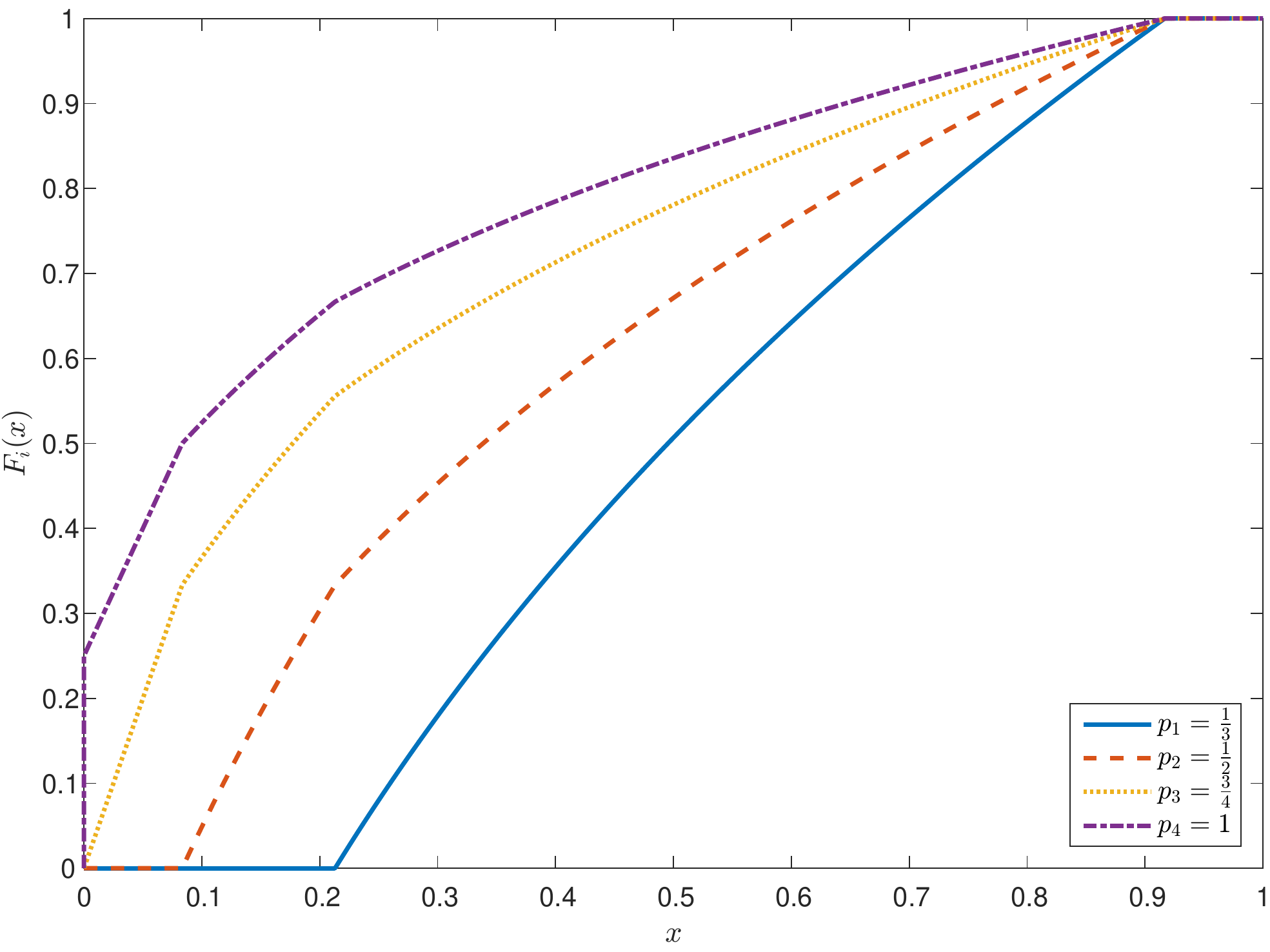}}		
			~\subfloat[\label{fig:examPDF}]{\includegraphics[width=0.5\textwidth]{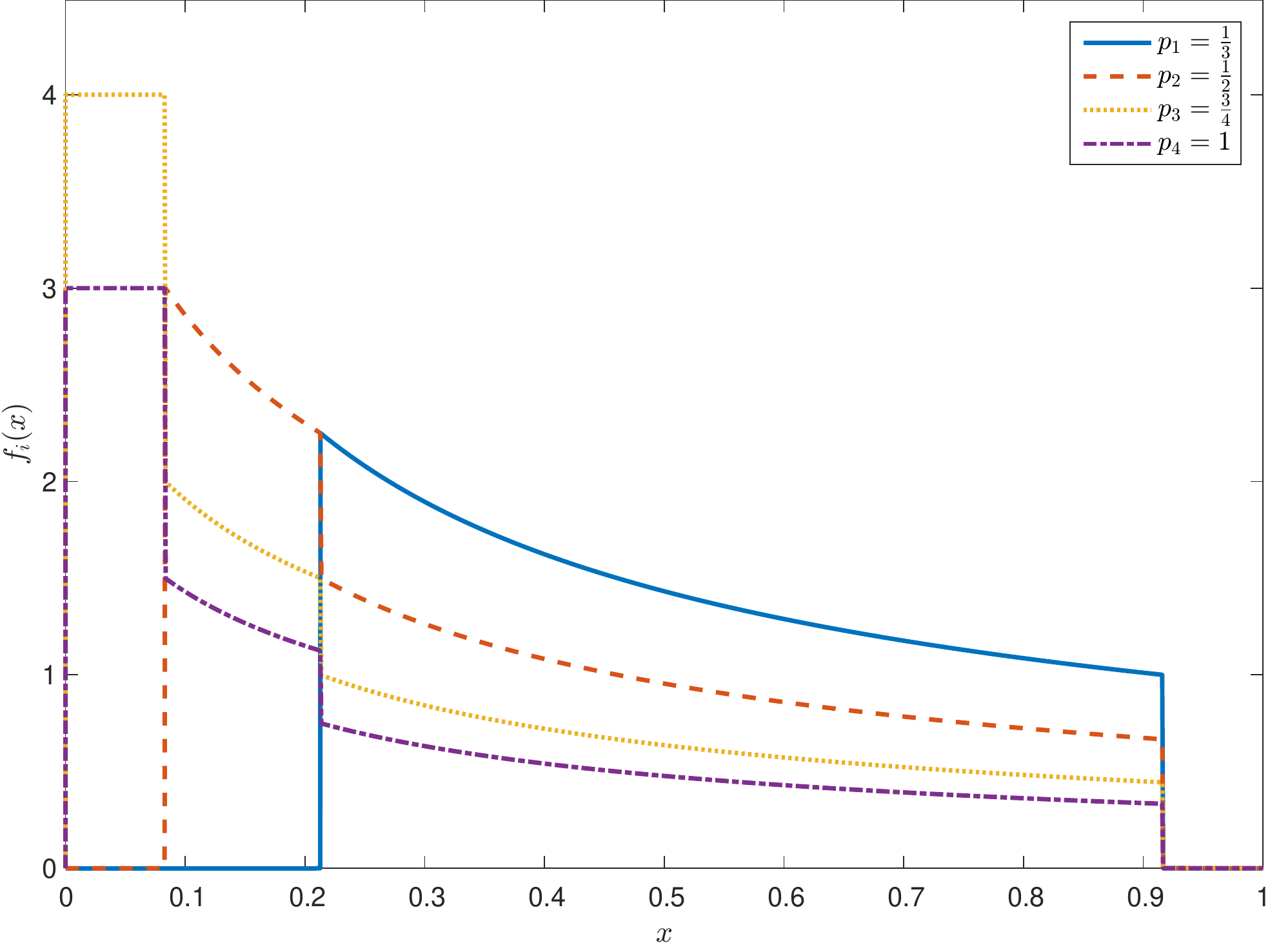}}
			\par
		\end{centering}	
		\centering{}\caption{The CDFs (a) and the PDFs (b) when $p_{1}=\frac{1}{3}$, $p_{2}=\frac{1}{2}$, $p_{3}=\frac{3}{4}$ and $p_{4}=1$}\label{fig:exam} 
	\end{figure}
	
	A graphical illustration of the bidders' CDFs (\ref{fig:examCDF}) and PDFs (\ref{fig:examPDF}) can be found in figure \ref{fig:exam}. 	
\end{customexm}

\section*{Proof of Theorem \ref{thm:sbg} \label{sec:sbg}}

Let $p_1,\dots, p_n$ be the announced participation probabilities, and let $p'_r < p_r$ be bidder $r$ real participation probability.
 For every $i \neq r$, 
 bidder $i$'s expected profit from bid $x$ is:
 \begin{equation*}
 \pi_{i}(x)=\left(p'_{r}F_{r}(x)+1-p'_{r}\right)\prod_{j=1;j\neq i,r}^{n}\left(p_{j}F_{j}(x)+1-p_{j}\right).
 \end{equation*}
 The values of this function change according to the relation between $r$, $i$ and $x$. The following lemmata offer some insights regarding the optimal strategy for a bidder. 
 \begin{lem}
 	Let $p_1,\dots, p_n$ be the announced participation probabilities, and let $p'_r < p_r$ be bidder $r$ real participation probability. For every $i > r$, 
 	the expected profit of bidder $i$ from bidding $x\in\left[\underline{s}_{k},\underline{s}_{k-1}\right)$, for $k\leq i$ and $k>r$, is:
 	$$
 	\pi_{i}(x)=\frac{1-p'_{r}}{1-p_{r}}\left(\lambda+x\right)-x=\frac{p_{r}-p'_{r}}{1-p_{r}}\left(\lambda+x\right)+\lambda.
 	$$
 	\label{lem:sup1}
 \end{lem}

\begin{proof}	
	The expected profit of bidder $i$ from bidding $x\in\left[\underline{s}_{k},\underline{s}_{k-1}\right)$, for $k\leq i$ and $k>r$, is:	
	\begin{equation}
	\begin{array}{rl}
	\pi_{i}\left(x\right)= & \prod\limits_{j=1;j\neq i,r}^{n}\left(p_{j}F_{j}\left(x\right)+1-p_{j}\right)\left(p'_{r}F_{r}\left(x\right)+1-p'_{r}\right)-x\\
	= & \prod\limits_{j=1;j\neq r}^{k-1}\left(p_{j}F_{j}\left(x\right)+1-p_{j}\right)\prod\limits_{j=1;j\neq i}^{n}\left(p_{j}F_{j}\left(x\right)+1-p_{j}\right)\left(p'_{r}F_{r}\left(x\right)+1-p'_{r}\right)-x\\
	= & \prod\limits_{j=1;j\neq r}^{k-1}\left(1-p_{j}\right)\prod\limits_{j=1;j\neq i}^{n}\left(H_{k}\left(x\right)\right)\left(1-p'_{r}\right)-x\\
	= & \prod\limits_{j=1;j\neq r}^{k-1}\left(1-p_{j}\right)H_{k}\left(x\right)^{n-k}\left(1-p'_{r}\right)-x\\
	= & \prod\limits_{j=1;j\neq r}^{k-1}\left(1-p_{j}\right)\left(\frac{\lambda+x}{\prod_{j=0}^{k-1}\left(1-p_{j}\right)}\right)\left(1-p'_{r}\right)-x\\
	= & \frac{1-p'_{r}}{1-p_{r}}\left(\lambda+x\right)-x\\
	= & \left(\frac{p_{r}-p'_{r}}{1-p_{r}}\right)\left(\lambda+x\right)+\lambda.
	\end{array}
	\end{equation}	
\end{proof}
In this case, this is larger than $\lambda$,
and grows with the bid, though the maximal bid (due to the fact that
$k > r$) is $\underline{s}_{r}$.

When either $i < r$ or $ i > r$, bidder can bid in the support of both bidders. 
\begin{lem}
	Let $p_1,\dots, p_n$ be the announced participation probabilities, and let $p'_r$ be bidder $r$ real participation probability. For every $i \neq r$, 
	the expected profit of bidder $i$ from bidding $x\in\left[\underline{s}_{k},\underline{s}_{k-1}\right)$, such that $k\leq \min\{i,r\}$ is:
	$$
	\pi_{i}(x)=\frac{p_{r}-p'_{r}}{p_{r}}\left(\lambda+x\right)\left(\left(\frac{\prod_{j=1}^{k-1}\left(1-p_{j}\right)}{\lambda+x}\right)^{\frac{1}{n-k}}-1\right)+\lambda.
	$$
\label{lem:sup2}	
\end{lem}

\begin{proof}	
	The expected profit of bidder $i$ from bidding $x\in\left[\underline{s}_{k},\underline{s}_{k-1}\right)$, such that $k\leq \min\{i,r\}$ is:	
	\begin{equation}
	\begin{array}{rl}
	\pi_{i}\left(x\right)= & \prod\limits_{j=1;j\neq i,r}^{n}\left(p_{j}F_{j}\left(x\right)+1-p_{j}\right)\left(p'_{r}F_{r}\left(x\right)+1-p'_{r}\right)-x\\
	= & \prod\limits_{j=1}^{k-1}\left(p_{j}F_{j}\left(x\right)+1-p_{j}\right)\prod\limits_{j=k;j\neq i,r}^{n}\left(p_{j}F_{j}\left(x\right)+1-p_{j}\right)\left(p'_{r}F_{r}\left(x\right)+1-p'_{r}\right)-x\\
	= & \prod\limits_{j=1}^{k-1}\left(1-p_{j}\right)\prod\limits_{j=k;j\neq i,r}^{n}\left(H_{k}\left(x\right)\right)\left(p'_{r}F_{r}\left(x\right)+1-p'_{r}\right)-x\\
	= & \prod\limits_{j=1}^{k-1}\left(1-p_{j}\right)H_{k}\left(x\right)^{n-k-1}\left(p'_{r}\frac{H_{k}\left(x\right)+p_{r}-1}{p_{r}}+1-p'_{r}\right)-x\\
	= & \prod\limits_{j=1}^{k-1}\left(1-p_{j}\right)H_{k}\left(x\right)^{n-k-1}\left(\frac{p'_{r}}{p_{r}}H_{k}\left(x\right)+\frac{p_{r}-p'_{r}}{p_{r}}\right)-x\\
	= & \prod\limits_{j=1}^{k-1}\left(1-p_{j}\right)H_{k}\left(x\right)^{n-k}\frac{p'_{r}}{p_{r}}+\prod\limits_{j=1}^{k-1}\left(1-p_{j}\right)H_{k}\left(x\right)^{n-k-1}\frac{p_{r}-p'_{r}}{p_{r}}-x\\
	
	= & \prod\limits_{j=1}^{k-1}\left(1-p_{j}\right)\left(\frac{\lambda+x}{\prod_{j=0}^{k-1}\left(1-p_{j}\right)}\right)\frac{p'_{r}}{p_{r}}+\prod\limits_{j=1}^{k-1}\left(1-p_{j}\right)H_{k}\left(x\right)^{n-k-1}\frac{p_{r}-p'_{r}}{p_{r}}-x\\
	= & \frac{p'_{r}}{p_{r}}\left(\lambda+x\right)+\prod\limits_{\begin{array}{c}
		j=1\end{array}}^{k-1}\left(1-p_{j}\right)\left(\frac{\lambda+x}{\prod_{j=0}^{k-1}\left(1-p_{j}\right)}\right)^{\frac{n-k-1}{n-k}}\frac{p_{r}-p'_{r}}{p_{r}}-x\\
	= & \frac{p'_{r}}{p_{r}}\left(\lambda+x\right)+\prod\limits\limits_{\begin{array}{c}
		j=1\end{array}}^{k-1}\left(1-p_{j}\right)^{\frac{1}{n-k}}\left(\lambda+x\right)^{\frac{n-k-1}{n-k}}\frac{p_{r}-p'_{r}}{p_{r}}-x\\
	= & \left(\lambda+x\right)\left(\frac{p'_{r}}{p_{r}}+\frac{p_{r}-p'_{r}}{p_{r}}\left(\frac{\lambda+x}{\prod_{j=0}^{k-1}\left(1-p_{j}\right)}\right)^{\frac{-1}{n-k}}\right)-x\\
	
	= & \left(\lambda+x\right)\frac{p_{r}-p'_{r}}{p_{r}}\left(\left(\frac{\lambda+x}{\prod_{j=0}^{k-1}\left(1-p_{j}\right)}\right)^{\frac{-1}{n-k}}-1\right)+\lambda.
	\end{array}
	\end{equation}	
\end{proof}

Since $x<\underline{s}_{k-1}$, this means that $\lambda+x<\left(1-p_{k-1}\right)^{n-k}\prod_{j=0}^{k-1}\left(1-p_{j}\right)$,
hence $\left(\frac{\lambda+x}{\prod_{j=0}^{k-1}\left(1-p_{j}\right)}\right)^{\frac{-1}{n-k}}>1$; again, this is an increase over $\lambda$, the position without sabotaging.

\begin{lem}	
	Let $p_1,\dots, p_n$ be the announced participation probabilities, and let $p'_r < p_r$ be bidder $r$ real participation probability. For every $i \neq r$ and for every $x\in \left[0, \underline{s}_i\right)$, there exists $x' \in \left[\underline{s}_i, 1 - \lambda \right]$, such that $\pi_i\left(x\right) \leq \pi_i\left(x'\right)$.\label{lem:outSupport}
\end{lem}

\begin{proof}
	When $i<r$, bidder $i$ can bid in bidder $r$'s support and not
	in its own support, i.e., bidder $i$ can bid $x\in\left[\underline{s}_{k},\underline{s}_{k-1}\right)$
	for $k\leq r$ and $i<k$. We first note that bidder $i$ can
	bid $\underline{s}_{i}$ in its support with expected utility of:
	\begin{equation}
	\pi_{i}\left(\underline{s}_{i}\right)=\left(1-p_{i}\right)^{n-i-1}\prod_{j=0}^{i-1}\left(1-p_{j}\right)\frac{p_{r}-p'_{r}}{p_{r}}p_{i}+\lambda.
	\end{equation}	
	When $x$ is in bidder $r$'s support and not in bidder $i$'s support:	
	\begin{equation}
	\begin{array}{rl}
	\pi_{i}\left(x\right)= & \prod\limits_{j=1;j\neq i,r}^{n}\left(p_{j}F_{j}\left(x\right)+1-p_{j}\right)\left(p'_{r}F_{r}\left(x\right)+1-p'_{r}\right)-x\\
	= & \prod\limits_{j=1;j\neq i}^{k-1}\left(p_{j}F_{j}\left(x\right)+1-p_{j}\right)\prod\limits_{j=k;j\neq r}^{n}\left(p_{j}F_{j}\left(x\right)+1-p_{j}\right)\left(p'_{r}F_{r}\left(x\right)+1-p'_{r}\right)-x\\
	= & \prod\limits_{j=1;j\neq i}^{k-1}\left(1-p_{j}\right)\prod\limits_{j=k;j\neq r}^{n}\left(H_{k}\left(x\right)\right)\left(\frac{p'_{r}}{p_{r}}\left(H_{k}\left(x\right)-1\right)+1\right)-x\\
	= & \frac{\prod_{j=0}^{k-1}\left(1-p_{j}\right)}{1-p_{i}}H_{k}\left(x\right)^{n-k}\left(\frac{p'_{r}}{p_{r}}\left(H_{k}\left(x\right)-1\right)+1\right)-x\\
	= & \frac{\lambda+x}{1-p_{i}}\left(\frac{p'_{r}}{p_{r}}\left(\frac{\lambda+x}{\prod_{j=0}^{k-1}\left(1-p_{j}\right)}\right)^{\frac{1}{n-k}}-\frac{p'_{r}}{p_{r}}+1\right)-x\\
	= & \frac{\lambda+x}{1-p_{i}}\left(\frac{p'_{r}}{p_{r}}\left(\frac{\lambda+x}{\prod_{j=0}^{k-1}\left(1-p_{j}\right)}\right)^{\frac{1}{n-k}}-\frac{p'_{r}}{p_{r}}+p_{i}\right)+\lambda.
	\end{array}
	\end{equation}
	If $p_{i}\geq\frac{p'_{r}}{p_{r}}$ then $\pi_{i}\left(x\right)$
	monotonic grows with $x$, and maximized when $x=\underline{s}_{i}$;
	While if $p_{i}<\frac{p'_{r}}{p_{r}}$, differentiating $\pi_{i}\left(x\right)$
	twice with respect to $x$ gives:	
	\begin{equation}
	\frac{\partial ^{2}}{\partial x^{2}}\pi_{i}\left(x\right)=\frac{p'_{r}}{p_{r}}\frac{1}{1-p_{i}}\left(1+\frac{1}{n-k}\right)\left(\frac{1}{n-k}\right)\left(\frac{\lambda+x}{\prod_{j=0}^{k-1}\left(1-p_{j}\right)}\right)^{\frac{1}{n-k}-1}\prod_{j=0}^{k-1}\left(1-p_{j}\right)^{\frac{-1}{n-k}},
	\end{equation}
	that is, for every $x\in\left[\underline{s}_{k},\underline{s}_{k-1}\right)$:
	$\frac{\partial ^{2}}{\partial x^{2}}\pi_{i}\left(x\right)\geq0$,
	hence, the extremum is a minimum point and
	the utility maximized when either $x=\underline{s}_{k-1}$ or $x=\underline{s}_{k}$.
	
	When bidding $\underline{s}_{k}$ the expected utility is:
	\begin{equation}
	\begin{array}{rl}
	\pi_{i}\left(\underline{s}_{k}\right)= & \frac{\prod_{j=0}^{k-1}\left(1-p_{j}\right)\left(1-p_{k}\right)^{n-k}}{1-p_{i}}\left(\frac{p'_{r}}{p_{r}}\left(1-p_{k}\right)-\frac{p'_{r}}{p_{r}}+p_{i}\right)+\lambda\\
	= & \frac{\prod_{j=0}^{k-1}\left(1-p_{j}\right)\left(1-p_{k}\right)^{n-k}}{1-p_{i}}\left(p_{i}-\frac{p'_{r}}{p_{r}}p_{k}\right)+\lambda.
	\end{array}
	\end{equation}
	If $p_{i}-\frac{p'_{r}}{p_{r}}p_{k}<0$ then $\pi_{i}\left(\underline{s}_{k}\right)<\lambda\leq\pi_{i}\left(\underline{s}_{i}\right)$.
	
	Otherwise, if $p_{i}-\frac{p'_{r}}{p_{r}}p_{k}\geq0$, since $i<k$ it holds that $p_{i}\leq p_{k}$, and therefore $\prod_{j=0}^{k-1}\left(1-p_{j}\right)\left(1-p_{k}\right)^{n-k}\leq\prod_{j=0}^{i-1}\left(1-p_{j}\right)\left(1-p_{i}\right)^{n-i}$
	and	
	\begin{equation}
	\begin{array}{rl}
	\pi_{i}\left(\underline{s}_{k}\right)\leq & \frac{\prod_{j=0}^{i-1}\left(1-p_{j}\right)\left(1-p_{i}\right)^{n-i}}{1-p_{i}}\left(p_{i}-\frac{p'_{r}}{p_{r}}p_{k}\right)+\lambda\\
	= & \prod_{j=0}^{i-1}\left(1-p_{j}\right)\left(1-p_{i}\right)^{n-i-1}\left(p_{i}-\frac{p'_{r}}{p_{r}}p_{k}\right)+\lambda\\
	\leq & \prod_{j=0}^{i-1}\left(1-p_{j}\right)\left(1-p_{i}\right)^{n-i-1}\left(p_{i}-\frac{p'_{r}}{p_{r}}p_{i}\right)+\lambda\\
	= & \prod_{j=0}^{i-1}\left(1-p_{j}\right)\left(1-p_{i}\right)^{n-i-1}\frac{p_{r}-p'_{r}}{p_{r}}p_{i}+\lambda\\
	= & \pi_{i}\left(\underline{s}_{i}\right).
	\end{array}
	\end{equation}	
	
	Similarly, $\pi_{i}\left(\underline{s}_{k-1}\right)\leq\pi_{i}\left(\underline{s}_{i}\right)$.
	That is, for every $x\in\left[\underline{s}_{k},\underline{s}_{k-1}\right)$:
	$\pi_{i}\left(x\right)\leq\pi_{i}\left(\underline{s}_{i}\right)$.
	Therefore, bidder $i$ has no incentive to bid outside of its support
	and in bidder $r$'s support.
	
	When $x$ is outside of the support of both bidders, i.e. $x\in\left[\underline{s}_{k},\underline{s}_{k-1}\right)$
	for $k>i$ and $k>r$:	
	\begin{equation}
	\begin{array}{rl}
	\pi_{i}\left(x\right)= & \prod\limits_{j=1;j\neq i,r}^{n}\left(p_{j}F_{j}\left(x\right)+1-p_{j}\right)\left(p'_{r}F_{r}\left(x\right)+1-p'_{r}\right)-x\\
	= & \prod\limits_{j=1;j\neq i,r}^{k-1}\left(p_{j}F_{j}\left(x\right)+1-p_{j}\right)\prod\limits_{j=k}^{n}\left(p_{j}F_{j}\left(x\right)+1-p_{j}\right)\left(p'_{r}F_{r}\left(x\right)+1-p'_{r}\right)-x\\
	= & \prod\limits_{j=1;j\neq i,r}^{k-1}\left(1-p_{j}\right)\prod\limits_{j=k}^{n}\left(H_{k}\left(x\right)\right)\left(1-p'_{r}\right)-x\\
	= & \frac{\prod_{j=0}^{k-1}\left(1-p_{j}\right)}{1-p_{i}}H_{k}\left(x\right)^{n-k+1}\frac{1-p'_{r}}{1-p_{r}}-x\\
	= & \frac{\prod_{j=0}^{k-1}\left(1-p_{j}\right)}{1-p_{i}}\frac{1-p'_{r}}{1-p_{r}}\left(\frac{\lambda+x}{\prod_{j=0}^{k-1}\left(1-p_{j}\right)}\right)^{\frac{n-k+1}{n-k}}-x\\
	= & \frac{\prod_{j=0}^{k-1}\left(1-p_{j}\right)^{\frac{-1}{n-k}}}{1-p_{i}}\frac{1-p'_{r}}{1-p_{r}}\left(\lambda+x\right)^{\frac{n-k+1}{n-k}}-x\\
	= & \left(\lambda+x\right)\left(\frac{1-p'_{r}}{\left(1-p_{i}\right)\left(1-p_{r}\right)}\left(\frac{\lambda+x}{\prod_{j=0}^{k-1}\left(1-p_{j}\right)}\right)^{\frac{1}{n-k}}-1\right)+\lambda.
	\end{array}
	\end{equation}
	Differentiating the above twice with respect to $x$ gives:	
	\begin{equation}
	\frac{\partial ^{2}}{\partial x^{2}}\pi_{i}\left(x\right)=\frac{1-p'_{r}}{\left(1-p_{i}\right)\left(1-p_{r}\right)}\left(1-\frac{1}{n-k}\right)\frac{1}{n-k}\left(\frac{\lambda+x}{\prod_{j=0}^{k-1}\left(1-p_{j}\right)}\right)^{\frac{1}{n-k}-1}\prod_{j=0}^{k-1}\left(1-p_{j}\right)^{\frac{-1}{n-k}}.
	\end{equation}
	Again, for every $x\in\left[\underline{s}_{k},\underline{s}_{k-1}\right)$: 
	$\frac{\partial ^{2}}{\partial x^{2}}\pi_{i}\left(x\right)\geq0$,
	hence, the extremum is a minimum point and
	the utility maximized when either $x=\underline{s}_{k-1}$ or $x=\underline{s}_{k}$.

	When bidding $\underline{s}_{k}$ the
	expected utility is:
	\begin{equation}
	\begin{array}{rl}
	\pi_{i}\left(\underline{s}_{k}\right)= & \left(1-p_{k}\right)^{n-k}\left(\frac{1-p'_{r}}{\left(1-p_{i}\right)\left(1-p_{r}\right)}\left(\frac{\left(1-p_{k}\right)^{n-k}\prod_{j=0}^{k-1}\left(1-p_{j}\right)}{\prod_{j=0}^{k-1}\left(1-p_{j}\right)}\right)^{\frac{1}{n-k}}-1\right)\prod_{j=0}^{k-1}\left(1-p_{j}\right)+\lambda\\
	= & \left(1-p_{k}\right)^{n-k}\left(\frac{1-p'_{r}}{\left(1-p_{i}\right)\left(1-p_{r}\right)}\left(1-p_{k}\right)-1\right)\prod_{j=0}^{k-1}\left(1-p_{j}\right)+\lambda.
	\end{array}
	\end{equation}
	If $\frac{1-p'_{r}}{\left(1-p_{i}\right)\left(1-p_{r}\right)}\left(1-p_{k}\right)-1\leq0$,
	then $\pi_{i}\left(\underline{s}_{k}\right)\leq\lambda\leq\pi_{i}\left(\underline{s}_{r}\right)$.
	
	Otherwise, $\frac{1-p'_{r}}{\left(1-p_{i}\right)\left(1-p_{r}\right)}\left(1-p_{k}\right)-1>0$. In this case we separate the two sub-cases; if $r < i $, the expected utility for bidder $i$ from bidding $\underline{s}_{r}$
	in its support is:	
	\begin{equation}
	\pi_{i}\left(\underline{s}_{r}\right)=\left(p_{r}-p'_{r}\right)\left(1-p_{r}\right)^{n-r-1}\prod_{j=0}^{r-1}\left(1-p_{j}\right)+\lambda
	\end{equation}
	since $k>i$ we have that $\frac{1-p_{k}}{1-p_{i}}\leq1$ and since
	$k>r$ we have that \newline $\prod_{j=0}^{k-1}\left(1-p_{j}\right)\left(1-p_{k}\right)^{n-k}\leq\prod_{j=0}^{r-1}\left(1-p_{j}\right)\left(1-p_{r}\right)^{n-r}$,
	hence--	
	\begin{equation}
	\begin{array}{rl}
	\pi_{i}\left(\underline{s}_{k}\right)= & \left(\frac{1-p'_{r}}{\left(1-p_{i}\right)\left(1-p_{r}\right)}\left(1-p_{k}\right)-1\right)\left(1-p_{k}\right)^{n-k}\prod_{j=0}^{k-1}\left(1-p_{j}\right)+\lambda\\
	\leq & \left(\frac{1-p'_{r}}{1-p_{r}}-1\right)\left(1-p_{k}\right)^{n-k}\prod_{j=0}^{k-1}\left(1-p_{j}\right)+\lambda\\
	\leq & \left(\frac{1-p'_{r}}{1-p_{r}}-1\right)\left(1-p_{k}\right)^{n-r}\prod_{j=0}^{r-1}\left(1-p_{j}\right)+\lambda\\
	\leq & \left(\frac{1-p'_{r}}{1-p_{r}}-1\right)\left(1-p_{r}\right)^{n-r}\prod_{j=0}^{r-1}\left(1-p_{j}\right)+\lambda\\
	= & \frac{p_{r}-p'_{r}}{1-p_{r}}\left(1-p_{r}\right)^{n-r}\prod_{j=0}^{r-1}\left(1-p_{j}\right)+\lambda\\
	= & \left(p_{r}-p'_{r}\right)\left(1-p_{r}\right)^{n-r-1}\prod_{j=0}^{r-1}\left(1-p_{j}\right)+\lambda\\
	= & \pi_{i}\left(\underline{s}_{r}\right).
	\end{array}
	\end{equation}
	Similarly $\pi_{i}\left(\underline{s}_{k-1}\right)\leq\pi_{i}\left(\underline{s}_{r}\right)$,
	that is, for every $x\in\left[\underline{s}_{k},\underline{s}_{k-1}\right)$:
	$\pi_{i}\left(x\right)\leq\pi_{i}\left(\underline{s}_{r}\right)$.
	And in a similar way if $r>i$, for every $k>r$, it holds that both $\pi_i \left( \underline{s}_{k} \right) \leq \pi_i \left(\underline{s}_{i} \right)$ and $ \pi_i \left(\underline{s}_{k-1} \right) \leq \pi_i \left(\underline{s}_{i} \right)$, thus for every $x\in\left[\underline{s}_{k},\underline{s}_{k-1}\right)$:
	 $\pi_{i}\left(x\right)\leq\pi_{i}\left(\underline{s}_{i}\right)$.
	Therefore, even if bidder $i$ sabotaged bidder $r$, bidder $i$ has
	no incentive to deviate from its support.
\end{proof}

Armed with Lemmata~\ref{lem:sup1},~\ref{lem:sup2} and~\ref{lem:outSupport} we may proceed to prove Theorem~\ref{thm:sbg}.
\begin{customthm}{\ref{thm:sbg}}
		Let $p_1,\dots, p_n$ be the announced participation probabilities, and let $p'_r < p_r$ be bidder $r$ real participation probability. 
		For every $i \neq r$, Algorithm~\ref{alg:opt} finds the optimal bid for bidder $i$.
\end{customthm}
\begin{proof}
	Lemmata~\ref{lem:sup1},~\ref{lem:sup2} and~\ref{lem:outSupport}
	tell us that the optimal bid is in $\left[\underline{s}_{\min\{i,r\}}, \underline{s}_0 \right]$.
	
Now, in order find the optimal bid in $\left[\underline{s}_{k},\underline{s}_{k-1}\right]$ for $k \leq \min\{i,r\}$,
we differentiate $\pi_i$ with respect to $x$:
\begin{equation}
\frac{\partial }{\partial x}\pi_{i}\left(x\right)=\frac{p_{r}-p'_{r}}{p_{r}}\left(1-\frac{1}{n-k}\right)\left(\frac{\lambda+x}{\prod_{j=1}^{k-1}\left(1-p_{j}\right)}\right)^{\frac{-1}{n-k}}-\frac{p_{r}-p'_{r}}{p_{r}}.
\end{equation}
That is, $\pi_{i}\left(x\right)$ maximized when 
\begin{equation}
x=\left(1-\frac{1}{n-k}\right)^{n-k}\prod_{j=1}^{k-1}\left(1-p_{j}\right)-\lambda.
\end{equation}
Now, $x\in\left[\underline{s}_{k},\underline{s}_{k-1}\right]$, and it holds that
\begin{equation}
\begin{array}{cc}
\left(1-\frac{1}{n-k}\right)^{n-k}\prod_{j=1}^{k-1}\left(1-p_{j}\right)-\lambda\in\left[\underline{s}_{k},\underline{s}_{k-1}\right] & \iff\\
\frac{1}{n-k}\in\left[p_{k-1},p_{k}\right].
\end{array}
\end{equation}
Thus, the optimal bid for bidder $i$ in $\left[\underline{s}_{k},\underline{s}_{k-1}\right]$,
depends if either $\frac{1}{n-k}\in\left[p_{k-1},p_{k}\right]$, $\frac{1}{n-k}< p_{k-1}$
or $\frac{1}{n-k}>p_{k}$. 
\begin{itemize}
	\item If $\frac{1}{n-k}\in\left[p_{k-1},p_{k}\right]$ then the optimal
	bid in $\left[\underline{s}_{k},\underline{s}_{k-1}\right]$ is
	$$x=\left(1-\frac{1}{n-k}\right)^{n-k}\prod_{j=1}^{k-1}\left(1-p_{j}\right)-\lambda$$
	and the expected profit is:	
	\begin{equation}
	\frac{1}{n-k}\left(1-\frac{1}{n-k}\right)^{n-k-1}\prod_{j=1}^{k-1}\left(1-p_{j}\right)\frac{p_{r}-p'_{r}}{p_{r}}+\lambda.
	\end{equation}	
	\item If $\frac{1}{n-k} < p_{k-1}$ then the optimal bid in $\left[\underline{s}_{k},\underline{s}_{k-1}\right]$
	is $x=\underline{s}_{k-1}$ and the expected profit is: 
	\begin{equation}
	p_{k-1}\left(1-p_{k-1}\right)^{n-k-1}\prod_{j=1}^{k-1}\left(1-p_{j}\right)\frac{p_{r}-p'_{r}}{p_{r}}+\lambda.
	\end{equation}	
	\item If $\frac{1}{n-k}>p_{k}$ then the optimal bid in $\left[\underline{s}_{k},\underline{s}_{k-1}\right]$
	is $x=\underline{s}_{k}$ and the expected profit is: 
	\begin{equation}
	p_{k}\left(1-p_{k}\right)^{n-k-1}\prod_{j=1}^{k-1}\left(1-p_{j}\right)\frac{p_{r}-p'_{r}}{p_{r}}+\lambda.
	\end{equation}	
\end{itemize}
Now, for every $k \leq \min\{i,r\}$ let $x_{i,k}$ be the optimal bid in $\left[\underline{s}_{k},\underline{s}_{k-1}\right]$,
 let $\pi_{i,k} = \pi_i\left(x_{i,k}\right)$ and let $j \in \arg\max_k\pi_{i,k}$.
As $x_{i,k}$ is the optima bid in $\left[\underline{s}_{k},\underline{s}_{k-1}\right]$, it follows that
$x_{i,j}$ is the optima bid in $\left[\underline{s}_{\min\{i,r\}}, \underline{s}_0 \right]$.
\end{proof}

\subsection*{Proof of Theorem~\ref{prop:expbid} \label{subsec:expbid}}

\begin{customthm}{\ref{prop:expbid}}
		The expected bid of every bidder is:
		\begin{equation}
		\mathbb{E}\left[bid\right] = \frac{1}{n p}\left(1-\lambda\left(1+p\left(n-1\right)\right)\right)
		\end{equation}
		and the variance of the bid is:
		\begin{equation}
		\mathrm{Var}\left[bid\right]= \frac{1-(1-p)^{2n-1}}{(2n-1)p}-\frac{\left(1-(1-p)^{n}\right)^{2}}{n^{2}p^{2}}.
		\end{equation}
		The expected bid and the variance are 
		neither monotonic in $n$ nor in $p$.
\end{customthm}
\begin{proof}

As this case is a particular instance of the general case presented
above, we can characterize the CDF for every bidder ---
\begin{equation}
F_{n,p}\left(x\right)=\begin{cases}
1 & x>1-\lambda\\
\frac{\left(\lambda+x\right)^{\frac{1}{n-1}}+p-1}{p}\,\, & x\in\left[0,1-\lambda\right]\\
0 & x<0
\end{cases}
\end{equation}
where $\lambda=\left(1-p\right)^{n-1}$, which implies that the bidders'
PDF is:
\begin{equation}
f_{n,p}\left(x\right)=\frac{\mathrm{d}}{\mathrm{d}x}F_{n,p}\left(x\right)=\begin{cases}
0 & x>1-\lambda\\
\frac{\left(\lambda+x\right)^{\frac{2-n}{n-1}}}{p\left(n-1\right)}\,\, & x\in\left[0,1-\lambda\right]\\
0 & x<0
\end{cases}.
\end{equation}
The expected bid of every bidder is:
\begin{equation}
\begin{array}{rl}
\mathbb{E}\left[bid\right]= & \int_{0}^{1-\lambda}x f_{n,p}\left(x\right)\,\mathrm{d}x\\
= & \frac{1}{n p}\left(1-\lambda\left(1+p\left(n-1\right)\right)\right)
\end{array}
\end{equation}
which is neither monotonic in $n$ nor in $p$.

The bid squared, in expectation, is:
\begin{equation}
\begin{array}{rl}
\mathbb{E}\left[bid^{2}\right]= & \int_{0}^{1-\lambda}x^{2} f_{n,p}\left(x\right)\,\mathrm{d}x\\
= & \left.\left(\frac{\left(\lambda+x\right)^{\frac{1}{n-1}}\lambda^{2}}{p}-2\frac{\left(\lambda+x\right)^{\frac{n}{n-1}}\lambda}{n p}+\frac{\left(\lambda+x\right)^{\frac{2n-1}{n-1}}}{\left(2n-1\right)p}\right)\right|_{0}^{1-\lambda}\\
= & \frac{1}{p}\left(\frac{1}{2n-1}-\frac{2}{n}\left(1-p\right)^{n-1}+\left(1-p\right)^{2n-2}-\left(1-p\right)^{2n-1}\frac{2\left(n-1\right)^{2}}{n\left(2n-1\right)}\right)
\end{array}
\end{equation}
hence, 
\begin{equation}
\begin{array}{rl}
Var\left(bid\right)= & \mathbb{E}\left[bid^{2}\right]-\mathbb{E}^{2}\left[bid\right]\\
= & \frac{1}{p}\left(\frac{1}{\left(2n-1\right)}-\frac{2}{n}\left(1-p\right)^{n-1}+\left(1-p\right)^{2n-2}-\left(1-p\right)^{2n-1}\frac{2\left(n-1\right)^{2}}{n\left(2n-1\right)}\right)\\
& -\left(\frac{1}{n p}\left(1-\left(1-p\right)^{n-1}\left(1+p\left(n-1\right)\right)\right)\right)^{2}\\
= & \frac{1-(1-p)^{2n-1}}{(2n-1)p}-\frac{\left(1-(1-p)^{n}\right)^{2}}{n^{2}p^{2}}
\end{array}
\end{equation}
which is, again, neither monotonic in $n$ nor in $p$.
\end{proof}
\subsection*{Proof of Theorem~\ref{prop:varbid} \label{subsec:varbid}}
\begin{customthm}{\ref{prop:varbid}}
	The variance of the bidder profit is:
	$$
	\mathrm{Var}\left[BP\right]= \frac{n-1}{n\left(2n-1\right)}-\frac{\left(1-p\right)^{n}}{n}+\left(p+\frac{1}{2n-1}\right)\left(1-p\right)^{2n-1}.
	$$			
	And the variance is monotonic increasing in $p$.
\end{customthm}

\begin{proof}
The squared of the bidder profit is: 
\begin{equation}
\begin{array}{rl}
\mathbb{E}\left[BP^{2}\right]= & \int_{\lambda-1}^{0}z^{2} g_{BP}\left(z\right)\,\partial z+\int_{\lambda}^{1}z^{2} g_{BP}\left(z\right)\,\mathrm{d} z\\
= & \int_{\lambda-1}^{0}z^{2}\frac{\left(\lambda-z\right)^{\frac{2-n}{n-1}}}{n-1}\left(1-\left(\lambda-z\right)\right)\,\partial z+\int_{\lambda}^{1}z^{2}\frac{\left(\lambda+1-z\right)^{\frac{1}{n-1}}}{n-1}\,\mathrm{d} z\\
= & \left.\left(\lambda-z\right)^{\frac{1}{n-1}}\left(\frac{\lambda(\lambda+2)\left(\lambda-z\right)}{n}+\frac{\left(\lambda-z\right)^{3}}{3n-2}-\lambda^{2}-\frac{(2\lambda+1)\left(\lambda-z\right)^{2}}{2n-1}\right)\right|_{\lambda-1}^{0}\\
 & \left. +\left(\lambda-z\right)^{\frac{1}{n-1}}\left(\left(\lambda+1-z\right)\left(\frac{2\lambda+2}{2n-1}-\frac{\lambda+1-z}{3n-2}\right)-\frac{(\lambda+1)^{2}}{n}\right)\right|_{\lambda}^{1}\\
= & \left(\lambda-z\right)^{\frac{1}{n-1}}\left(\frac{2\lambda\left(n-1\right)^{2}}{n\left(3n-2\right)\left(2n-1\right)}-\frac{2\left(n-1\right)^{2}}{n\left(2n-1\right)}\right)+\lambda^{2}\frac{n-1}{n}-\lambda\frac{2\left(n-1\right)}{n\left(2n-1\right)}+\frac{n-1}{\left(3n-2\right)\left(2n-1\right)}\\
 & -\left(1-p\right)^{n}\left(\frac{2\lambda^{2}\left(n-1\right)^{2}}{n\left(3n-2\right)\left(2n-1\right)}+\frac{2\lambda\left(n-1\right)}{n\left(2n-1\right)}+\frac{1}{n}\right)\\
 & +\frac{\lambda^{2}}{n}+\frac{2\lambda\left(n-1\right)}{n\left(2n-1\right)}+\frac{2\left(n-1\right)^{2}}{n\left(3n-2\right)\left(2n-1\right)}\\
= & \frac{n-1}{n\left(2n-1\right)}-\frac{\left(1-p\right)^{n}}{n}+\left(1-p\right)^{2n-2}-\frac{2\left(n-1\right)}{2n-1}\left(1-p\right)^{2n-1}.
\end{array}
\end{equation}
Hence the variance of the profit of every bidder is:
\begin{equation}
\begin{array}{rl}
\mathrm{Var}\left[BP\right]= & \mathbb{E}\left[BP^{2}\right]-\mathbb{E}^{2}\left[BP\right]\\
= & \frac{n-1}{n\left(2n-1\right)}-\frac{\left(1-p\right)^{n}}{n}+\left(p+\frac{1}{2n-1}\right)\left(1-p\right)^{2n-1}
\end{array}
\end{equation}

Differentiating the above with respect to $p$ gives:
\begin{equation}
\frac{\partial}{\partial p}Var\left(BP\right)=(1-p)^{n-1}-2n\cdot p\left(1-p\right)^{2n-2}.
\end{equation}
Now, 
\begin{equation}
\begin{array}{rlc}
\frac{\partial }{\partial p}\mathrm{Var}\left[BP\right] & \geq0 & \iff\\
(1-p)^{n-1}-2n\cdot p\left(1-p\right)^{2n-2} & \geq0 & \iff\\
p\left(1-p\right)^{n-1} & \leq\frac{1}{2n}
\end{array}
\end{equation}
since $p\left(1-p\right)^{n-1}$ maximized when $p=\frac{1}{n}$,

\begin{equation}
\begin{array}{rlc}
\frac{\partial }{\partial p}\mathrm{Var}\left[BP\right] & \geq0 & \iff\\
\left(1-\frac{1}{n}\right)^{n-1} & \leq\frac{1}{2}
\end{array}
\end{equation}
which holds for every $n\geq2$. Therefore for every $n\geq2$ the variance
of the bidder profit increases with $p$.

\end{proof}

\subsection*{Proof of Theorem~\ref{prop:varmax} \label{subsec:varmax}}
	\begin{customthm}{\ref{prop:varmax}}
		The variance of the auctioneer in the max profit model is:
		\begin{equation}
		\begin{array}{rl}
		\mathrm{Var}\left[AP\right]= & \left(1-p\right)^{2n-2}-\frac{2n\left(1-p\right)^{n-1}}{2n-1}+\frac{n}{3n-2}-\frac{2\left(n-1\right)^{2}\left(1-p\right)^{3n-2}}{\left(3n-2\right)\left(2n-1\right)}\\
		& -\left(\frac{n}{2n-1}+\frac{n-1}{2n-1}\left(1-p\right)^{2n-1}-\left(1-p\right)^{n-1}\right)^{2}.
		\end{array}
		\end{equation}
\end{customthm}

\begin{proof}
	The profit squared, in expectation, is:
	\begin{equation}
	\begin{array}{rl}
	\mathbb{E}\left[AP^{2}\right]= & \int_{0}^{1-\lambda}x^{2}\cdot\frac{n}{n-1}\left(\lambda+x\right)^{\frac{1}{n-1}}\,\mathrm{d}x\\
	= &\left. \left(\lambda+x\right)^{\frac{n}{n-1}}\left(\lambda^{2}-\lambda\frac{2n\left(\lambda+x\right)}{2n-1}+\frac{n\left(\lambda+x\right)^{2}}{3n-2}\right)\right|_{0}^{1-\lambda}\\
	= & \left(1-p\right)^{2n-2}-\frac{2n\left(1-p\right)^{n-1}}{2n-1}+\frac{n}{3n-2}-\frac{2\left(n-1\right)^{2}\left(1-p\right)^{3n-2}}{\left(3n-2\right)\left(2n-1\right)}.
	\end{array}
	\end{equation}	
	Hence, the variance is:
	\begin{equation}
	\begin{array}{rl}
	\mathrm{Var}\left[AP\right]= & \left(1-p\right)^{2n-2}-\frac{2n\left(1-p\right)^{n-1}}{2n-1}+\frac{n}{3n-2}-\frac{2\left(n-1\right)^{2}\left(1-p\right)^{3n-2}}{\left(3n-2\right)\left(2n-1\right)}\\
	& -\left(\frac{n}{2n-1}+\frac{n-1}{2n-1}\left(1-p\right)^{2n-1}-\left(1-p\right)^{n-1}\right)^{2}.
	\end{array}
	\end{equation}
\end{proof}

\end{document}